\documentclass[12pt]{article}

\usepackage{latexsym,amsmath,amscd,amssymb,graphics,cite}
\usepackage{enumerate}
\usepackage{mathtools}
\usepackage{graphicx}
\usepackage{framed}
\usepackage[colorlinks]{hyperref}
\usepackage{url}
\usepackage{cite}
\usepackage{cancel}

\usepackage[all]{xy}

\makeatletter

\@addtoreset{figure}{section}
\def\thefigure{\thesection.\@arabic\c@figure}
\def\fps@figure{h, t}
\@addtoreset{table}{bsection}
\def\thetable{\thesection.\@arabic\c@table}
\def\fps@table{h, t}
\@addtoreset{equation}{section}

\makeatother

\newtheorem{theorem}{Theorem}

\newtheorem{remark}[theorem]{Remark}

\numberwithin{theorem}{section}
\newenvironment{proof}[1][Proof]{\textbf{#1.} }{\ \rule{0.5em}{0.5em}}

\def\be{\begin{equation}}
\def\ee{\end{equation}}
\def\bea{\begin{eqnarray}}
\def\eea{\end{eqnarray}}
\def\ba{\begin{array}}
\def\ea{\end{array}}

\def\bOm{\boldsymbol{\Omega}}

\let\<\langle
\let \>\rangle

\newcommand{\rem}[1]{}

\newcommand{\de}{\delta}

\newcommand{\bGam}{\boldsymbol{\Gamma}}
\newcommand{\bom}{\boldsymbol{\omega}}
\newcommand{\bgam}{\boldsymbol{\gamma}}
\newcommand{\bsigma}{\boldsymbol{\Sigma}}
\newcommand{\bpsi}{\boldsymbol{\Psi}}

\newcommand{\bmu}{\boldsymbol{\mu}}

\newcommand{\bpi}{\boldsymbol{\pi}}

\newcommand{\bchi}{\boldsymbol{\chi}} 

\newcommand{\pp}[2]{\frac{\partial #1}{\partial #2}}

\newcommand{\dede}[2]{\frac{\delta #1}{\delta #2}}
\newcommand{\prt}{\partial}

\newcommand{\lp}{\left(}
\newcommand{\rp}{\right)}

\newcommand{\todo}[1]{\vspace{5 mm}\par \noindent
\framebox{\begin{minipage}[c]{0.95 \textwidth}
\tt #1 \end{minipage}}\vspace{5 mm}\par}

\title{Geometric theory of flexible and expandable tubes conveying fluid: equations, solutions and shock waves}
\author{Fran\c{c}ois Gay-Balmaz$^1$ and 
Vakhtang Putkaradze$^2$ \vspace{0.2cm}\\
\small $^1$ CNRS \& LMD-IPSL, Ecole Normale Sup\'erieure de Paris, France\\
\small $^2$Department of Mathematical and Statistical Sciences\\
\small University of Alberta, Edmonton, AB   T6G 2G1 Canada
\\
}
\date{}

\begin{document}
\maketitle

\begin{abstract}
We present a theory for the three-dimensional evolution of tubes with expandable walls conveying fluid. Our theory can accommodate arbitrary deformations of the tube, arbitrary elasticity of the walls, and both compressible and incompressible flows inside the tube. We also present the theory of propagation of shock waves in such tubes and derive the conservation laws and Rankine-Hugoniot conditions in arbitrary spatial configuration of the tubes, and compute several examples of particular solutions. The theory is derived from a variational treatment of Cosserat rod theory extended to incorporate expandable walls and moving flow inside the tube. The results presented here are useful for biological flows and industrial applications involving high speed motion of gas in flexible tubes.
\end{abstract}


\section{Introduction}
Tubes with flexible walls conveying fluid are encountered frequently in nature and exhibit complex behavior due to the interaction of fluid and elastic walls dynamics. The theory for such tubes is indispensable in biomedical applications, such as arterial flows \cite{PeLu1998,QuTuVe2000,FoLaQu2003,StWaJe2009,Tang-etal-2009} and lung flows \cite{ElKaSh1989,MaFl2010,Do2016}. Analytical studies for such flows are usually limited to cases when the centerline of the tube is straight, which restricts the utility of the models for practical applications. Of interest to this work is the mechanism of instability through the neck formation and self-sustaining flow pulsations suggested by \cite{Pe1992} and experimentally measured in \cite{KuMa1999}. While substantial progress in the analysis of the flow has been achieved so far, it was difficult to describe analytically the general dynamics of 3D deformations of the tube involving \emph{e.g.} the combination of shear, transversal deformation, and extensions.

 The numerical simulations of flow in channels with flexible walls, such as  undertaken in \cite{LuPe1998} for the case of a two-dimensional channel, yield a wealth of details of the velocity profile for the fluid, but have to restrict the dimensionality of the motion. For more informations about the application of collapsible tubes to biological flows, we refer the reader to the reviews \cite{GrJe2004,HeHa2011}. 

On the other hand, studies involving non-trivial dynamics of the centerline have a long history in the context of engineering applications also loosely called the ``garden hose instabilities". One of the earliest works on the subject is \cite{AsHa1950}. Benjamin \cite{Be1961a,Be1961b} was perhaps the first to formulate a quantitative theory for the 2D dynamics  of {\it initially straight tubes} by considering a linked chain of tubes conveying fluids and using an augmented Hamilton principle of critical action that takes into account the momentum of the jet leaving the tube. A continuum equation for the linear disturbances was then derived as the limit of the discrete system. This linearized equation for the initially straight tubes was further studied by Gregory and Pa\"idoussis \cite{GrPa1966a}. 

These initial developments formed the basis for further stability analysis for finite, initially straight tubes \cite{Pa1970,PaIs1974,Pa1998,ShMi2001,DoLa2002,PaLi1993,Pa2004,AkIvKoNe2013,AkGeNe2015,AkGeNe2016}, which showed  a reasonable agreement with experimentally observed onset of the instability \cite{GrPa1966b,Pa1998,KuSa2005,CaCr2009,Cr-etal-2012}. Models treating nonlinear deflection of the centerline  were also considered in \cite{SeLiPa1994,Pa2004,MoPa2009,GaPaAm2013}, and the compressible (acoustic) effects in the flowing fluid in \cite{Zh2008}.  Alternatively,  a more detailed 3D theory of motion for the flow was developed in \cite{BeGoTa2010} and extended in \cite{RiPe2015}. That theory was based on a modification of the Cosserat rod treatment for the description of elastic dynamics of the tube, while keeping the cross-section of the tube constant and orthogonal to the centerline. In particular, \cite{RiPe2015} analyzes several non-straight configurations, such as a tube deformed from its straight state by gravity,  both from the point of view of linear stability and nonlinear behavior. Unfortunately, this Cosserat-based theory cannot easily incorporate the effects of the cross-sectional changes in the dynamics, either prescribed (e.g. a tube with variable cross-section in the initial state) or dynamically occurring.  

Several authors have discussed the instability from the point of view of the {\it follower force approach}, which treats the system as an elastic beam, ignoring the fluid motion,  with a force that is always tangent to the end of the tube,  modeling the effect of the jet leaving the nozzle \cite{BoRoSa2002}. However, once the length of the tube becomes large enough, the validity of the follower force approach has been questioned, see  \cite{El2005} for a lively and thorough discussion.
For the history of the development of this problem in the Soviet/Russian literature, we refer the reader to the monograph \cite{Sv1987} (still only available in Russian). The developments in Russian literature  have often proceeded in parallel with the English literature counterpart, see \emph{e.g.} \cite{Mo1965,Mu1965,Il1969,Anni1970,VoGr1973,Sv1978,Do1979,Ch1984,So2005,AmAl2015}. The theory of stability of initially curved pipes have also been considered because its importance for practical applications, such as pipelines.
The equations of motion for the theory were derived using the balance of elastic forces arising from the tube deformation and  fluid forces acting on the tube when the fluid is moving along a curved line in space.
In the western literature, we shall mention the earlier work \cite{Ch1972}, followed by subsequent studies \cite{MiPaVa1988a,MiPaVa1988b,DuRo1992} which developed the theory suited for both extensible and inextensible tubes and discussed the finite-element method realization of the problem. We shall also mention the works \cite{DoMo1976,AiGi1990} deriving a variational approach for the planar motions of initially circular tubes, although the effect of curved fluid motion was still introduced as an extra force through the Lagrange-d'Alembert principle.  In the Soviet/Russian literature,   \cite{Sv1978} developed a rod-based theory of oscillations and \cite{So2005} considered an improved treatment of forces acting on the tubes. Most of the work has been geared towards the understanding of the planar cases with in-plane vibrations as the simplest and most practically relevant situations (still, however, leading to quite complex formulas). 

Finally, it is worth discussing the shell theory approach to the dynamics of fluid-filled pipes \cite{AmPePa1999a,AmPePa1999b}. These theories are usually derived for relatively small (but perhaps nonlinear) deviations of an initially cylindrical shell from its equilibrium position. They are capable of describing quite complex modes of oscillations along the circumference of the tubes, but are too cumbersome to extend to tubes with moving centerline. 
We also refer the reader to the classical treatise \cite{Pa1998,Pa2004} on fluid-structure interactions, where the developments of theory and experiment are discussed in great details. 

In spite of considerable progress achieved so far,  there is still much room for improvement in the theoretical understanding of the problem. In particular, we believe that the derivation of the theory based on the balance of forces is not variational and the approximations of certain terms tend to break down the intrinsic variational structure of the problem. In contrast, the theory of flexible tubes conveying fluid as developed in \cite{FGBPu2014,FGBPu2015} is truly variational and all the forces acting on the tube and the fluid, as well as the boundary forces, are derived automatically from a variational principle issued from Hamilton's principle. More importantly, without a variational approach it is very difficult (and perhaps impossible)  to extend the previous theory to accurately take into account the changes in the cross-sectional area of the tube, also called the collapsible tube case.

In many previous works, the effects of cross-sectional changes have  been considered through the quasi-static approximation: if $A(t,s)$ is the local cross-section area, and $u(t,s)$ is the local velocity of the fluid, with $t$ being the time and $s$ the coordinate along the tube,  then the quasi-static assumption states that  $uA=$ const, see the corresponding discussion in \cite{FGBGePu2018}. The use of the quasi-static approximation is due to the fact that an accurate description of conservation law for a moving tube is difficult to obtain in the traditional approach when the tube undergoes  motion of rotations and translations. Interestingly, most works on modelling flows in a tube with flexible walls and \emph{a straight centerline with no rotation} do not use this approximation and employ the correct conservation law. The connection between models derived in the present paper and that work, most often used in the context of hemodynamics  \cite{PeLu1998,KoMa1999,JuHe2007,StWaJe2009,Tang-etal-2009,HeHa2011}, see also recent review \cite{Se2016}, is further discussed in \S\ref{sec:comparison}. This problem has been addressed by  the authors for general tube motion and prescribed cross-sectional area as a function of tube's deformations \cite{FGBPu2014,FGBPu2015}, where a geometrically exact setting for dealing with a variable cross-section was developed and studied, showing the important effects of the cross-sectional changes on both linear and nonlinear dynamics, as well as boundary effects. In these works by the authors, the cross-sectional area was considered to be a given function of the local deformation of the rod, in other words, the wall of the tube did not have its own dynamics, which is a commonly used approximation in the field. This theory has been further developed in \cite{FGBGePu2018}, where it was shown that thanks to the geometric approach and variational principle used, the linear stability of an arbitrary helical tube reduces to the analysis of a  constant matrix. In addition, the variational derivation for that type of fluid-structure interaction allowed for the development of variational numerical schemes \cite{FGBPu2016}. 

In this paper, we considerably extend the previous results and construct a complete theory of fluid-conveying tubes capable of incorporating the wall dynamics as well as  arbitrary deformations of the tube with general elastic and inertia properties, as long as the flow inside the tube can be approximated by a one-dimensional ``plug flow". Our theory covers both compressible and incompressible flows and can describe effects like propagation of shock waves in the tubes in arbitrary tube geometries, which is important for industrial applications. Mathematically, our theory considers the left symmetry-reduced dynamics of the elastic tube in the convective representation (i.e. in the frame of reference of the tube), the fluid motion in the spatial (Eulerian) representation in the tube's frame of reference, the Euler-Lagrange equations of the wall motion, and either the holonomic constraint expressing the incompressibility of the fluid, or the equation for mass conservation for compressible flows.

\section{Exact geometric theory for flexible tubes conveying compressible fluid}\label{sec:general}

In this section, we first quickly review the Lagrangian variational formulation for geometrically exact rods without fluid motion. Then, we extend this variational formulation to incorporate the motion of the compressible fluid inside the tube and the motion of the wall. To achieve this goal we need to first identify the infinite dimensional configuration manifold of the system as well as the convective and spatial variables, together with their relation with the Lagrangian variables.
Our model is then obtained by an application of the Hamilton principle, reformulated in convective variables for the tube and in spatial variables for the fluid. The rigorous justification of the variational approach will be given later in Section \ref{Geometric_approach}. We shall then study the conditions the shock solutions must satisfy, thereby extending the classical Rankine-Hugoniot to the case of shock waves moving in a tube that can  freely deform in the three dimensional space and with expandable cross-section.

\subsection{Background for geometrically exact rod theory}

In this paragraph we briefly review the theory of geometrically exact rods following the approach developed, on the Hamiltonian side, in \cite{SiMaKr1988} and subsequently in the Lagrangian framework more appropriate to this article in \cite{HoPu2009,ElGBHoPuRa2010}. A more comprehensive introduction is contained in \cite{FGBGePu2018} to which we refer the reader for the details. The geometrically exact theory is based on the original work \cite{CoCo1909}.

The configuration of the rod deforming in the ambient space $ \mathbb{R}  ^3 $ is defined by specifying the position of its line of centroids by means of a map $\mathbf{r} (t,s)\in \mathbb{R}^3$, and by giving the orientation of the cross-section at that point. Here $t$ is the time and $s\in [0,L]$ is a parameter along the strand that does not need to be arclength. The orientation of the cross-section is given by a moving basis $\{ \mathbf{e} _i(t,s) \mid i=1,2,3\}$ attached to the cross section relative to a fixed frame $\{ \mathbf{E} _i \mid i=1,2,3\}$. The moving basis is described by an orthogonal transformation $ \Lambda (t,s) \in SO(3)$ such that $ \mathbf{e} _i (t,s)= \Lambda (t,s) \mathbf{E} _i $.
We interpret the maps $ \Lambda (t,s)$ and $\mathbf{r}(t,s)$ as a curve $t\mapsto (\Lambda(t),\mathbf{r}(t))\in G$ in the infinite dimensional Lie group $G= \mathcal{F} ([0,L], SO(3) \times \mathbb{R} ^3  )$ of $SO(3) \times \mathbb{R} ^3$-valued smooth maps defined on $[0,L]$. The Lie group $G$ is the configuration manifold for the geometrically exact rod.

Following Hamilton's principle, given a Lagrangian function
\[
\mathsf{L}=\mathsf{L}( \Lambda , \dot {\Lambda }, \mathbf{r} , \dot{\mathbf{r}}):TG \rightarrow \mathbb{R},
\]
defined on the tangent bundle $TG$ of the configuration Lie group $G$, the equations of motion are the Euler-Lagrange equations obtained by the critical action principle
\begin{equation}\label{HP_rod} 
\delta \int_0^T\mathsf{L}( \Lambda , \dot {\Lambda }, \mathbf{r} , \dot{\mathbf{r}})\mbox{d}t=0,
\end{equation} 
for arbitrary variations $\delta\Lambda$ and $\delta  \mathbf{r} $ vanishing at $t=0,T$.
It turns out that the Lagrangian of geometrically exact rods can be exclusively expressed in terms of the convective variables 
\begin{equation}\label{def_conv_var}
\begin{aligned}
&\bgam=\Lambda^{-1} \dot{\mathbf{r} }\,,  &\quad 
&\bom=\Lambda^{-1} \dot\Lambda\,,
\\ 
&\bGam=\Lambda^{-1} \mathbf{r} '\,, &\quad
&\bOm=\Lambda^{-1} \Lambda'\,,
\end{aligned} 
\end{equation} 
see \cite{SiMaKr1988}, where $ \boldsymbol{\gamma} (t), \boldsymbol{\omega} (t)\in \mathcal{F} ([0,L], \mathbb{R}  ^3 )$ are the linear and angular convective velocities and $ \boldsymbol{\Gamma} (t), \boldsymbol{\Omega}(t)\in \mathcal{F} ([0,L], \mathbb{R}  ^3 )$ are the linear and angular convective strains. This gives rise to a Lagrangian $\ell=\ell( \boldsymbol{\omega} , \boldsymbol{\gamma} , \boldsymbol{\Omega} , \boldsymbol{\Gamma} ): \mathcal{F}([0,L], \mathbb{R}  ^3 )^4  \rightarrow \mathbb{R}  $ written exclusively in terms of convective variables. For the moment, we leave the Lagrangian function unspecified, we will give its explicit expression later in \S\ref{Def_Lagrangian} for the case of fluid-conveying tubes.

The equations of motion in convective description are obtained by writing the critical action principle \eqref{HP_rod} in terms of the Lagrangian $\ell$. This is accomplished by computing the constrained variations of $\boldsymbol{\omega} , \boldsymbol{\gamma} , \boldsymbol{\Omega} , \boldsymbol{\Gamma} $ induced by the free variations $ \delta \Lambda , \delta\mathbf{r} $ via the definitions \eqref{def_conv_var}. 
We find
\begin{align} 
& \delta \boldsymbol{\omega} = \frac{\partial \boldsymbol{\Sigma} }{\partial t} +\boldsymbol{\omega} \times \boldsymbol{\Sigma} , \qquad \delta \boldsymbol{\gamma} = \frac{\partial \boldsymbol{\psi} }{\partial t} + \boldsymbol{\gamma} \times \boldsymbol{\Sigma} + \boldsymbol{\omega} \times \boldsymbol{\Psi} 
\label{delta1} 
\\
& \delta \boldsymbol{\Omega} = \frac{\partial \boldsymbol{\Sigma} }{\partial s} +\boldsymbol{\Omega} \times \boldsymbol{\Sigma} , \qquad \delta \boldsymbol{\Gamma} = \frac{\partial \boldsymbol{\psi} }{\partial s} + \boldsymbol{\Gamma} \times \boldsymbol{\Sigma} + \boldsymbol{\Omega} \times \boldsymbol{\Psi},
\label{delta2} 
\end{align} 
where $ \bsigma(t,s)=\Lambda(t,s)^{-1} \de \Lambda(t,s)\in \mathbb{R}  ^3 $ and $\bpsi (t,s)= \Lambda(t,s)^{-1} \de \mathbf{r} (t,s) \in \mathbb{R}  ^3 $ are arbitrary functions vanishing at $t=0,T$.
Hamilton's principle \eqref{HP_rod} induces the variational principle
\begin{equation}\label{Reduced_HP_rod} 
\delta \int_0^T\ell( \boldsymbol{\omega} , \boldsymbol{\gamma} , \boldsymbol{\Omega} , \boldsymbol{\Gamma} )\mbox{d}t=0,
\end{equation} with respect to the constrained variations \eqref{delta1}, \eqref{delta2}, which in turn yields the reduced Euler-Lagrange equations
\begin{equation}
\left\lbrace\begin{array}{l}
\displaystyle\lp \prt_t + \bom\times\rp\dede{\ell}{\bom}+\bgam\times\dede{\ell}{\bgam} +\lp\prt_s + \bOm\times\rp\dede{\ell}{\bOm} +\bGam\times \dede{\ell}{\bGam} =0\\
\displaystyle\lp \prt_t + \bom\times\rp\dede{\ell}{\bgam} + \lp\prt_s + \bOm\times\rp  \dede{\ell}{\bGam}=0,
\end{array}\right.
\end{equation} 
together with the boundary conditions
\begin{equation}\label{BC_rod} 
\left.\frac{\delta \ell}{\delta \boldsymbol{\Omega} }\right |_{s=0,L}=0\quad \left.\frac{\delta \ell}{\delta \boldsymbol{\Gamma } }\right |_{s=0,L}=0.
\end{equation} 
If one of the extremity (say $s=0$) of the rod is kept fixed, \emph{i.e.},  $ \mathbf{r} (t,0)=\mathbf{r} _0 $, $ \Lambda (t,0)= \Lambda _0 $ for all $t$, then only the boundary condition at $s=L$ arises above.

From their definition \eqref{def_conv_var}, the convective variables verify the compatibility conditions
\begin{equation}\label{compat_rod} 
 \partial _t \boldsymbol{\Omega} = \boldsymbol{\omega} \times \boldsymbol{\Omega} +\partial _s  \boldsymbol{\omega}\quad\text{and}\quad  \partial _t \boldsymbol{\Gamma} + \boldsymbol{\omega} \times \boldsymbol{\Gamma} = \partial _s \boldsymbol{\gamma} + \boldsymbol{\Omega} \times \boldsymbol{\gamma}.
\end{equation}

\begin{remark}[Lagrangian reduction by symmetry]\rm  The process of passing from the Lagrangian (or material) representation in terms of $ (\Lambda , \dot{ \Lambda }, \mathbf{r} , \dot{ \mathbf{r} })$ with variational principle \eqref{HP_rod} to the convective representation in terms of $(\boldsymbol{\omega} , \boldsymbol{\gamma} , \boldsymbol{\Omega} , \boldsymbol{\Gamma})$ with constrained variational principle \eqref{Reduced_HP_rod} can be understood via a Lagrangian reduction process by symmetries. It has been carried out in \cite{ElGBHoPuRa2010} and is based on the affine Euler-Poincar\'e reduction theory of \cite{GBRa2009}.  We shall review this approach later, in Section \ref{Geometric_approach}, before extending it to fluid-conveying tubes.
\end{remark}

\begin{remark}[On the functional form of the Lagrangian] \rm Note that in all our theoretical considerations we will keep the Lagrangian in the general form, as we are interested in the symmetry-reduction approach to the fully three dimensional problem, rather than in the derivation of the equations of motion in a particular reduced setting,  \emph{e.g.}, restricted to two dimensions, straight line, \emph{etc}. We believe that such an approach based on Lagrangian mechanics yields the simplest possible treatment of the elastic, three-dimensional deformation of the tube. We shall present class of Lagrangian functions in \S\ref{Def_Lagrangian}. In particular, we shall consider a particular class of quadratic elastic energy, see \eqref{e_rod} later, allowing for further substantial simplification for the motion when constrained in two dimensions. 
\end{remark}

\subsection{Definition of the configuration space}

We now incorporate the motion of the fluid inside the tube and the  motion of the wall of the tube  by  extending  the geometrically exact framework.
Recall that the geometrically exact rod (without fluid) consists of a left invariant system and, therefore, can be written in terms of convective variables. On the other hand, the fluid is a right invariant system, naturally written in terms of spatial variables.
The coupling of these two systems therefore yields the interesting concept of a system involving both convective and spatial variables but whose left and right invariances are broken by the coupling constraint.

In addition to the rod variables $(\Lambda , \mathbf{r}) \in \mathcal{F} ([0,L], SO(3) \times \mathbb{R}  ^3 )$ considered above, the configuration manifold for the fluid-conveying tube also contains the Lagrangian description of the fluid. It is easier to start by defining the back-to-label map, which is an embedding $ \beta :[ 0,L] \rightarrow \mathbb{R}  $, assigning to a current fluid label particle $s \in [0,L]$ located at $ \mathbf{r} (s) $ in the tube, its Lagrangian label $s_0 \in \mathbb{R}  $. Its inverse $ \varphi := \beta ^{-1} : \beta([0,L]) \subset \mathbb{R}  \rightarrow [0,L]$ gives the current configuration of the fluid in the tube. A time dependent curve of such maps thus describes the fluid motion in the tube, \emph{i.e.}, 
\[
s= \varphi (t,s_0), \quad s \in [0,L].
\]
We now include the motion of the wall of the tube, as a reaction to the fluid motion and pressure. In order to incorporate this effect in the simplest possible case, let us consider the tube radius $R(t,s)$ to be a free variable. In this case, the Lagrangian depends on  $R$, as well as on its time and space derivatives $\dot R$ and $R'$, respectively. If we assume that $R$ can lie on an interval $I_R$, for example, $I_R=\mathbb{R}_+$ (the set of positive numbers), then 
the configuration manifold for the fluid-conveying tube is given by the infinite dimensional manifold
\begin{equation}\label{config_garden_hose} 
\mathcal{Q} :=\mathcal{F} \left( [0,L], SO(3) \times \mathbb{R} ^3 \times I_R  \right) \times \left\{ \varphi : \varphi ^{-1} [0,L]  \rightarrow [0,L]\mid \text{$\varphi $ diffeomorphism} \right\}.
\end{equation}
Note that the domain of definition of the fluid motion $s= \varphi (t,s_0)$ is time dependent, i.e., we have $ \varphi (t):[ a(t), b(t)] \rightarrow [0,L]$, for $\varphi (t, a(t))=0$ and $\varphi (t,b(t))=L$.
The time dependent interval $[a(t),b(t)]$ contains the labels of all the fluid particles that are present in the tube at time $t$.

\subsection{Definition of the Lagrangian}\label{Def_Lagrangian}

Let us now turn our attention to the derivation of the Lagrangian of the fluid-conveying geometrically exact tube. Although we present below a particular expression of the Lagrangian in terms of the variables $(\bom, \bgam,\bOm, \bGam,R, \dot R,S)$, all the equations of motion will be valid for general Lagrangians.

\paragraph{Kinetic energy.} The kinetic energy of the elastic rod is the function $K_{\rm rod}$ given by
\[
K_{\rm rod}= \frac{1}{2} \int_0^L\left( \alpha | \bgam|^2 + a \dot{R}^2 + \mathbb{I}(R) \bom \cdot \bom \right)| \boldsymbol{\Gamma} | \mbox{d}s,
\] 
where $\alpha$ is the linear density of the tube and $\mathbb{I}(R)$ is the  local moment of inertia of the tube.  The term $ \frac{1}{2}a  \dot{R}^2$ describes the kinetic energy of the radial motion of the tube.

We now derive the total kinetic energy of the fluid. In material representation, the total velocity of the fluid particle with label $s_0$ is given by
\begin{equation}\label{velocity_equalities}
\begin{aligned} 
\frac{d}{dt} \mathbf{r}(t ,\varphi (t,s_0))&= \partial _t \mathbf{r} (t, \varphi (t,s_0))+ \partial _s \mathbf{r} (t, \varphi (t,s_0)) \partial _t \varphi (t,s_0)\\
&=\partial _t \mathbf{r} (t, \varphi (t,s_0))+ \partial _s \mathbf{r} (t, \varphi (t,s_0)) u(t,\varphi (t,s_0)),
\end{aligned}
\end{equation}  
where the Eulerian velocity is defined by
\begin{equation}\label{Eulerian_velocity_u}
u(t,s)=\left(  \partial _t \varphi \circ \varphi ^{-1} \right) (t,s), \quad s \in [0,L].
\end{equation}
Therefore, the total kinetic of the fluid reads
\[
K_{\rm fluid}= \frac{1}{2} \int_{ \varphi ^{-1} (0,t)} ^{ \varphi ^{-1} (L,t)}  \xi _0  (s_0) \left| \frac{d}{dt} \mathbf{r}(t ,\varphi (t,s_0))\right | ^2 \mbox{d}s_0,
\]
where the function $ \xi _0(s_0)$ denotes the mass density of the fluid per unit length in the material representation. It is related to the area $Q_0(s_0)$ and the mass density of the fluid per unit volume $ \rho _0(s_0)$ as $ \xi _0= \rho _0 Q_0$.
Using \eqref{velocity_equalities} together with the change of variables $s= \varphi (t,s_0)$, we can rewrite $K_{\rm fluid}$ as
\[
K_{\rm fluid}= \frac{1}{2} \int_{0} ^{ L}  ( \xi _0 \circ \varphi ^{-1} )\partial _s \varphi ^{-1} \left| \boldsymbol{\gamma} + \boldsymbol{\Gamma} u\right | ^2 \mbox{d}s,
\]
where
\begin{equation}\label{xi_advection}
\xi (t,s)= \left[ ( \xi _0  \circ  \varphi ^{-1} )\partial _s \varphi^{-1}  \right] (t,s)
\end{equation}
is the mass density per unit length in the Eulerian description. Note that we have the relation
\[
\xi (t,s)= \rho (t,s) Q(t,s),
\]
where $\rho (t,s)$ is the mass density of the fluid per unit volume, in units Mass/Length$^3$, and $Q(t,s)$ is the area of the tube's cross section, in units Length$^2$.  It is important to note that, while $\xi$, $\xi_0$ are related as in \eqref{xi_advection}, such a relation does not hold for $\rho$, $\rho_0$ and $Q$, $Q_0$, \emph{e.g.}, 
$Q (s,t) \neq  ( Q _0  \circ  \varphi ^{-1} )\partial _s \varphi^{-1} $. That relationship between $Q$ and $Q_0$ is only valid when the fluid inside the tube is incompressible, see \cite{FGBPu2014,FGBPu2015}.

\paragraph{Internal energy.} We assume that the thermal energy of the gas in the Lagrangian is described by the specific energy function $e(\rho,S)$, with $\rho$ being the mass density and $S$ being the  specific entropy. Recall the thermodynamic identities
\begin{equation} 
\mbox{d} e = - p \, \mbox{d} \left( \frac{1}{\rho} \right) + T \mbox{d} S  \quad \Rightarrow \quad p(\rho,S)=\rho^2 \pp{e}{\rho}(\rho , S)\, , \quad 
T(\rho,S) = \pp{e}{S} (\rho , S)\, , 
\label{thermodynamics} 
\end{equation} 
where $p(\rho,S)$ is the pressure and $T(\rho,S)$ is the temperature.  The total internal energy of the fluid is thus
\[
E_{\rm int}= \int_0^L \xi e( \rho , S) \mbox{d}s.
\]

\paragraph{Mass conservation.} We shall assume that the fluid fills the tube completely, and the fluid velocity at each given cross-section is aligned with the axis of the tube.
Since we are assuming a one-dimensional approximation for the fluid motion inside the tube,  the mass density per unit length $ \xi (t,s)$ has to verify
\begin{equation}\label{eq_Q} 
\xi = ( \xi _0 \circ \varphi ^{-1} )\partial _s \varphi ^{-1}.
\end{equation} 
This equation has been already used earlier in \eqref{xi_advection}, from which we deduce the conservation law
\begin{equation} 
\partial _t \xi +\partial _s( \xi u)=0.
\label{xi_cons}
\end{equation} 

The physical meaning of $Q$ can be understood through the fact that $Q \mbox{d} s$ is the volume filled by the fluid in the tube for parameter interval $[s,s+\mbox{d} s]$.
Since $s$ is not necessarily taken to be the arc length, we have $Q=A|\bGam|$, where $A$ is the area of the cross section. In general $Q$ is a given function of the tube's variables, i.e.,
\begin{equation}\label{Q_function} 
Q= Q( R, \boldsymbol{\Omega} , \boldsymbol{\Gamma} ).
\end{equation} 
This expression must be invariant under any re-parameterization $s \rightarrow \widetilde{s}(s)$. In order to explain the origin of this equation, let us approximate the tube's cross-section as an ellipse with semi-axes $a(t,s)$ and $b(t,s)$ at given values of $t$ and $s$. When the centerline  and the orientation of the cross-section of the tube deform,
 then the eccentricity of the cross-section changes depending on the deformation, so $b(t,s)= f(\bOm,\bGam) a(t,s)$, where $f(\bOm,\bGam)$ is a function known from experiments and $a(t,s)$ is the unknown variable.

For a tube with an initially round cross-section made out of uniform material, the cross-sectional area will deform depending on the bending, \emph{i.e.}, $\bOm \times \bchi$, and shear, 
\emph{i.e.}, $\bGam \times  \bchi$, where $\bchi$ is the axis of the tube in reference configuration, usually taken to be $\bchi=\mathbf{E}_1$. For such tubes we have $f=1$ for $\bOm=\mathbf{0}$ and $\bGam=\bchi$.  

For more complex tubes, made out of non-uniform synthetic materials, or commonly encountered in biological applications like arterial flows, all the components of $\bOm$ and $\bGam$ enter into the function $f(\bOm,\bGam)$. Then, the area of the cross-section is given by 
\begin{equation}
A(a,\bOm,\bGam)=\pi a^2 f(\bOm,\bGam) \, . 
\label{Aeq_gen}
\end{equation}
The choice of $A=A(\bOm,\bGam)$ was taken in \cite{FGBPu2015,FGBPu2016,FGBGePu2018}. Such choice prevents the independent dynamics of the tube's wall and states that the cross-sectional area only depends on the deformation of the tube as an elastic rod. For the physical explanation of possible particular expressions of $A(\bOm,\bGam)$ we refer the reader to \cite{FGBGePu2018}. 

The simplest choice we will use for computations is $f \equiv 1$, so that the tube preserves its circular cross-section under deformations. We then take $a=R$ and define 
\begin{equation}\label{modelA}
Q(R,\boldsymbol{\Gamma})=A(R)|\boldsymbol{\Gamma}|=\pi R^2  | \boldsymbol{\Gamma} |.
\end{equation} 
We shall derive all the equations for a general expression \eqref{Aeq_gen} and only use the circular-tube approximation in the computations of the dynamics when the cross-sections at any $s$ is initially circular, the centerline remains straight, and there is no additional twist of the tube. For this model, all cross-sections of the tube remain circular by symmetry.

\paragraph{Elastic energy.} The potential energy due to elastic deformation is a function of $\bOm$, $\bGam$ and $R$. While the equations will be derived for an arbitrary potential energy,  we shall assume the simplest possible quadratic expression for the calculations, namely,
\begin{equation}
E_{\rm rod}=\frac{1}{2}\int_0^L\Big( \mathbb{J}\bOm\! \cdot  \!\bOm
+ \lambda(R) |\bGam- \boldsymbol{\chi} |^2 +2F(R,R',R'')\Big)  |\bGam| \mbox{d} s\,,
\label{e_rod} 
\end{equation} 
where $ \boldsymbol{\chi} \in \mathbb{R} ^3 $ is a fixed vector denoting the axis of the tube in the reference configuration, $ \mathbb{J}$ is a symmetric positive definite $3 \times 3$ matrix, which may depend on $R$, $R'$ and $R''$, and $\lambda (R)$ is the stretching rigidity of the tube. 
The stretching term proportional to $\lambda(R)$ can take the more general form $\mathbb{K} (\bGam-  \boldsymbol{\chi} ) \cdot (\bGam-  \boldsymbol{\chi} )$, where $\mathbb{K}$ is a $3 \times 3$ tensor.
The part of this expression for the elastic energy containing the first two terms in \eqref{e_rod} is commonly used for a Cosserat elastic rod, but more general functions of deformations $\bGam $ are possible. A particular case is a quadratic  function  of $\bGam$ leading to a linear dependence between stresses and strains. We have also introduced the elastic energy of wall $F(R,R',R'')$ which can be explicitly computed for simple elastic tubes. In general $F$ depends on higher derivatives, such as $R''$, however, in this paper, we shall use the simplest possible  approximation for the elastic energy of the wall, corresponding to the long-wavelength approximation in deformations of $R'$, leading to $F$ being only  dependent on $R$ and $R'$. The derivation of the equations in the case when $F$ depends on higher derivatives of $R$ is straightforward, however, obtaining numerical solutions in this case becomes challenging. We shall also note that our geometric approach is valid for an arbitrary dependence of the elastic energy  $E_{\rm rod}(\bOm,\bGam,R)$ on deformations.

\paragraph{Lagrangian.} From all the expressions given above, we obtain the Lagrangian of the fluid-conveying given by
\begin{equation}\label{total_Lagrangian}
\mathsf{L}=\mathsf{L} \big(  \Lambda , \dot {\Lambda }, \mathbf{r} , \dot{\mathbf{r}}, \varphi , \dot{ \varphi },R,\dot{R}\big) :T\mathcal{Q}  \rightarrow \mathbb{R},\quad \mathsf{L}=K_{\rm rod}+K_{\rm fluid}-E_{\rm rod} - E_{\rm int},
\end{equation}
and defined on the tangent bundle $T \mathcal{Q} $ of the configuration space $ \mathcal{Q} $, see \eqref{config_garden_hose}. 
Note that all the arguments of $\mathsf{L}$ are functions of $s$, so we don't need to include the spatial derivatives of $(\Lambda, \mathbf{r}, R)$ explicitly as variables in $\mathsf{L}$. These spatial derivatives appear explicitly in the expression of the integrand of the reduced Lagrangian \eqref{Lagrangian_fluid_tube} below. Assuming there is a uniform external pressure $p_{\rm ext}$ acting on the tube, the Lagrangian expressed 
in terms of the variables $\boldsymbol{\omega} ,\boldsymbol{\gamma} , \boldsymbol{\Omega}, \boldsymbol{\Gamma} ,u,  \xi , S $ reads
\begin{equation}\label{Lagrangian_fluid_tube}
\begin{aligned}
&\ell  ( \boldsymbol{\omega} ,\boldsymbol{\gamma} , \boldsymbol{\Omega} ,  \boldsymbol{\Gamma} ,u, \xi , S, R,\dot{R}) 
\\
&=  \int_0^L \Big[\Big(\frac{1}{2} \alpha | \bgam|^2 +  \frac{1}{2}\mathbb{I}(R) \bom\! \cdot\! \bom + \frac{1}{2} a  \dot R^2 - F(R,R') - \frac{1}{2} \mathbb{J}  \bOm \!\cdot \! \bOm\\
& \qquad \qquad  - \frac{1}{2} \lambda(R) |\bGam- \boldsymbol{\chi} |^2 \Big) | \bGam|
+ \frac{1}{2} \xi  \left| \boldsymbol{\gamma} + \boldsymbol{\Gamma} u\right | ^2
- \xi  e(\rho,S)  -p_{\rm ext} Q\Big] \mbox{d} s 
\\ 
&=: 
\int_0^L \Big[\ell_0( \boldsymbol{\omega} ,\boldsymbol{\gamma} , \boldsymbol{\Omega} ,  \boldsymbol{\Gamma} ,u, \xi ,R,\dot{R}, R') 
-\xi e(\rho,S) -p_{\rm ext} Q \Big]  \mbox{d} s \, , 
\end{aligned}
\end{equation}
where $ \rho $, in the term $\xi e(\rho,S)$ in the above formula, is defined in terms of the independent variables $ \xi ,\boldsymbol{\Omega} ,\boldsymbol{\Gamma} , R$ as
\begin{equation}\label{relation_xi_rho_Q} 
\rho := \frac{\xi }{Q( \boldsymbol{\Omega} , \boldsymbol{\Gamma} , R)}.
\end{equation} 
For convenience in further calculations, we also define the function
\begin{equation} {\fontsize{11pt}{9pt}\selectfont
f_0:= \frac{1}{2}\Big( \alpha | \bgam|^2 +  \mathbb{I}(R) \bom\! \cdot\! \bom +  a  \dot R^2 -2 F(R,R') -  \mathbb{J}  \bOm \!\cdot \! \bOm
-\lambda(R) |\bGam- \boldsymbol{\chi} |^2 \Big) \,.
\label{f0_def} }
\end{equation}
In Section \ref{Geometric_approach} we shall show that $\ell$ arises from the Lagrangian $\mathsf{L}:T \mathcal{Q} \rightarrow \mathbb{R}  $ by a reduction process by symmetry. We shall thus refer to $\ell$ as the symmetry-reduced Lagrangian. We have also denoted $\ell_0$ to be the part of the integrand of the Lagrangian related to just the tube dynamics, without the incorporation of the internal energy.
We shall perform all the derivations for an arbitrary Lagrangian function $\ell( \boldsymbol{\omega} ,\boldsymbol{\gamma} , \boldsymbol{\Omega} ,  \boldsymbol{\Gamma} ,u, \xi , S, R,\dot{R}) $ and shall study later particular solutions for the expression \eqref{Lagrangian_fluid_tube}. 

\subsection{Variational principle and equations of motion}

The equations of motion are obtained from the Hamilton principle applied to the Lagrangian \eqref{total_Lagrangian}, namely
\begin{equation}\label{HP_total} 
\delta \int_0^T\mathsf{L}( \Lambda , \dot {\Lambda }, \mathbf{r} , \dot{\mathbf{r}}, \varphi, \dot\varphi, R, \dot R)\mbox{d}t=0, 
\end{equation} 
for arbitrary variations $\delta\Lambda, \delta\mathbf{r}, \delta\varphi, \delta R$ vanishing at $t=0,T$. In terms of the symmetry reduced Lagrangian $\ell$, this variational principle becomes
\begin{equation} 
\de  \int_0^T\ell ( \boldsymbol{\omega} ,\boldsymbol{\gamma} , \boldsymbol{\Omega} ,  \boldsymbol{\Gamma} ,u, \xi , S, R,\dot{R}) \mbox{d} t =0 
\, , 
\label{min_action_gas} 
\end{equation} 
for variations \eqref{delta1}, \eqref{delta2}, and $\de u$, $\de \xi$ and $\de S$ computed as
\begin{align} 
\de u & =  \partial_t\eta + u \partial_s\eta - \eta \partial _su\label{delta_u}\\
\delta \xi &=- \partial _s ( \xi \eta )\label{delta_xi}\\
\delta S&= -\eta \partial _s S,\label{delta_S}
\end{align} 
where $\eta = \de \varphi \circ \varphi^{-1}$. Note that $\eta(t,s)$ is an arbitrary function vanishing at $t=0,T$. A lengthy computation yields the system 
\begin{equation}\label{system_ell} 
\left\lbrace\begin{array}{l}
\displaystyle\vspace{0.2cm}\frac{D}{Dt} \dede{\ell}{\bom}+\bgam\times\dede{\ell}{\bgam} +\frac{D}{Ds}  \dede{\ell}{\bOm}   +\bGam\times \dede{\ell}{\bGam}=0\\
\displaystyle\vspace{0.2cm}\frac{D}{Dt} \dede{\ell}{\bgam} +\frac{D}{Ds}  \dede{\ell}{\bGam} =0\\
\displaystyle\vspace{0.2cm}\prt_t\frac{\delta \ell}{\delta u}  + u \partial _s\frac{\delta \ell}{\delta u}+ 2  \frac{\delta \ell}{\delta u} \partial _s u = \xi \partial _s \frac{\delta \ell}{\delta \xi }- \frac{\delta \ell}{\delta S} \partial _s S   \\
\displaystyle\vspace{0.2cm} 
\partial_t \frac{\delta \ell }{\delta \dot R} - \frac{\delta \ell}{\delta R} =0 
\\
\displaystyle\vspace{0.2cm} \partial _t \boldsymbol{\Omega} = \boldsymbol{\Omega} \times \boldsymbol{\omega} +\partial _s  \boldsymbol{\omega}, \qquad  \partial _t \boldsymbol{\Gamma} + \boldsymbol{\omega} \times \boldsymbol{\Gamma} = \partial _s \boldsymbol{\gamma} + \boldsymbol{\Omega} \times \boldsymbol{\gamma}\\
\displaystyle \partial _t  \xi + \partial _s ( \xi u)=0, \qquad\partial _tS+ u \partial _s S=0,
\end{array}\right.
\end{equation}
where the symbols $  {\delta \ell}/{\delta \boldsymbol{\omega} },  {\delta \ell}/{\delta \boldsymbol{\Gamma} }, ... $ denote the functional derivatives of $\ell$ relative to the $L ^2 $ pairing, and we introduced the notations
\[
\frac{D}{Dt} =\partial_t + \boldsymbol{\omega}\times\quad\text{and}\quad\frac{D}{Ds} =\partial_s + \boldsymbol{\Omega}\times.
\]
Note that the first equation arises from the terms proportional to $\bsigma$ in the variation of the action functional and thus describes the conservation of angular momentum. The second equation arises from the terms proportional to $\boldsymbol{\psi}$ and describes the conservation of linear momentum.
The third equation is obtained by collecting the terms proportional to $\eta$ and describes the conservation of fluid momentum.  The fourth equation comes from collecting the terms proportional to $\de R$ and describes the elastic deformation of the walls due to the pressure. Finally, the last four equations arise from the four definitions $\bOm=\Lambda^{-1} \Lambda'$, $\bGam=\Lambda^{-1} \mathbf{r} '$, $\xi =  ( \xi _0 \circ \varphi ^{-1} )\partial _s \varphi ^{-1}$, and $S= S_0 \circ \varphi ^{-1}$.

For the choice
\[
\ell( \boldsymbol{\omega} ,\boldsymbol{\gamma} , \boldsymbol{\Omega} ,  \boldsymbol{\Gamma} ,u, \xi , S, R,\dot{R})=
\int_0^L \Big[\ell_0( \boldsymbol{\omega} ,\boldsymbol{\gamma} , \boldsymbol{\Omega} ,  \boldsymbol{\Gamma} ,u, \xi ,R,\dot{R}, R') 
-\xi e(\rho,S) - p_{\rm ext}Q \Big] \mbox{d} s \, ,
\]
the functional derivatives are computed as
\begin{equation}\label{functional_derivatives}
\begin{aligned}
\frac{\delta \ell}{\delta \boldsymbol{\Omega} }&= \frac{\partial  \ell_0}{\partial \boldsymbol{\Omega} } +  (p- p_{\rm ext}) \frac{\partial Q}{\partial \boldsymbol{\Omega} }\\
\frac{\delta \ell}{\delta \boldsymbol{\Gamma } }&= \frac{\partial \ell_0}{\partial \boldsymbol{\Gamma } } +  (p- p_{\rm ext}) \frac{\partial Q}{\partial \boldsymbol{\Gamma } }\\
\frac{\delta \ell}{\delta R }&= \frac{\partial \ell_0}{\partial R } - \partial _s\frac{\partial \ell_0}{\partial R' }+ \partial _ s^2 \frac{\partial \ell_0}{\partial R'' } +  (p- p_{\rm ext}) \frac{\partial Q}{\partial R }\\
\frac{\delta \ell}{\delta \xi }&= \frac{\partial \ell_0}{\partial \xi  } - e -\rho  \frac{\partial e}{\partial \rho }  ,
\end{aligned}
\end{equation}
where $p( \rho , S)= \rho ^2 \frac{\partial e}{\partial \rho }(\rho , S)$ is the pressure and $\partial \ell_0/\partial \boldsymbol{\Omega}, \partial \ell_0/\partial \boldsymbol{\Gamma} $, ... denote the ordinary partial derivatives of $\ell_0$,  whose explicit form can be directly computed from the expression of $\ell_0$ in \eqref{Lagrangian_fluid_tube}.

\begin{theorem}\label{vp_tube_fluid}  For the Lagrangian $\ell$ in \eqref{Lagrangian_fluid_tube}, the variational principle \eqref{min_action_gas} with constrained variations \eqref{delta1}, \eqref{delta2}, \eqref{delta_u}, \eqref{delta_xi}, \eqref{delta_S} yields the equations of motion
\begin{equation}\label{full_3D}{\fontsize{11pt}{9pt}\selectfont
\!\!\!\!\left\lbrace\!\begin{array}{l}
\displaystyle\vspace{0.2cm}\frac{D}{Dt} \pp{\ell_0}{\bom}+\bgam\times\pp{\ell_0}{\bgam} +\frac{D}{Ds}  \left( \pp{\ell_0}{\bOm} + (p-p_{\rm ext})\pp{Q}{\bOm} \right)   +\bGam\times 
 \left( 
\pp{\ell_0}{\bGam} + (p-p_{\rm ext})\pp{Q}{\bGam} 
\right) =0
\\
\displaystyle\vspace{0.2cm}\frac{D}{Dt} \pp{\ell_0}{\bgam} +\frac{D}{Ds} \left( \pp{\ell_0}{\bGam}+(p-p_{\rm ext}) \pp{Q}{\bGam}   \right)=0\\
\displaystyle\vspace{0.2cm}\prt_t\frac{ \partial  \ell_0}{ \partial  u}  + u \partial _s\frac{ \partial  \ell_0}{ \partial  u}+ 2  \frac{ \partial  \ell_0}{ \partial  u} \partial _s u = \xi \partial _s  \pp{\ell_0}{\xi}  -Q \partial_s p \\
\displaystyle\vspace{0.2cm} 
\partial_t \pp{\ell_0}{\dot R} - \partial^2_s \pp{\ell_0}{ R''} +  \partial_s \pp{\ell_0}{ R'} -  \pp{\ell_0}{ R}- (p-p_{\rm ext})\frac{\partial Q}{\partial R}=0 
\\
\displaystyle\vspace{0.2cm} \partial _t \boldsymbol{\Omega} = \boldsymbol{\Omega} \times \boldsymbol{\omega} +\partial _s  \boldsymbol{\omega}, \qquad  \partial _t \boldsymbol{\Gamma} + \boldsymbol{\omega} \times \boldsymbol{\Gamma} = \partial _s \boldsymbol{\gamma} + \boldsymbol{\Omega} \times \boldsymbol{\gamma}\\
\displaystyle \partial _t \xi + \partial _s (\xi u)=0, \qquad\partial _tS+ u \partial _s S=0
\end{array}\right.}
\end{equation}
together with appropriate boundary conditions enforcing vanishing of the variations of the boundary terms. 
\rem{ 
\begin{equation}\label{BC} 
\left. \dede{\ell}{\bOm} - \mu \pp{Q}{\bOm} \right|_{s=0,L} =0 \, , 
\quad 
\left. \dede{\ell}{\bGam} - \mu \pp{Q}{\bGam} \right|_{s=0,L}=0 \, , \quad 
\left. \dede{\ell}{u}u - \mu Q \right|_{s=0,L}=0
\end{equation} 
hold. If one of the extremity of the rod is fixed, say $s=0$, then the first two boundary conditions only arise at $s=L$. If the fluid velocity is prescribed at one extremity of the rod, say $s=0$, then the last boundary condition only arise at $s=L$.
} 
\end{theorem} 
\begin{proof}  The system is obtained by replacing the expression of the functional derivatives \eqref{functional_derivatives} in the system \eqref{system_ell} and  using the thermodynamics identities \eqref{thermodynamics}.\end{proof} 

\medskip

Equations \eqref{full_3D}  have to be solved as a system of nonlinear partial differential equations, since all equations are coupled. This can be seen, for example, from computing the derivative 
\begin{equation} 
\frac{ \partial \ell _0 }{ \partial  u} = \rho A | \boldsymbol{\Gamma} |  \big( \bgam + \bGam u \big)\cdot \bGam\, ,
\label{dldu}
\end{equation}
which appears in the the balance of fluid momentum, \emph{i.e.}, the third equation in \eqref{full_3D}.

\medskip

With the fluid momentum defined as
\[
m:= \frac{1}{\rho Q}\dede{\ell_0}{u} = \bGam \cdot \left( \bgam+ u \bGam \right)
\]
the third equation in \eqref{full_3D} can be simply written as
\begin{equation}
\label{meq}
\partial _t m+  \partial _s \left(mu  -  \pp{\ell_0}{\xi} \right)=-\frac{1}{\rho} \partial_s p \, , \quad  \, . 
\end{equation}
which is strongly reminiscent of the 1D gas dynamics. 
Also, in \eqref{meq} we have
\begin{equation} 
\label{dell_dxi} 
\frac{\partial  \ell_0}{ \partial  \xi }= \frac{1}{2}  \left| \boldsymbol{\gamma} + \boldsymbol{\Gamma} u\right | ^2
\end{equation} 
with the physical meaning of total velocity of the fluid particle, squared.

\medskip

For $Q(R,\bGam)=A(R) |\bGam|$, the equations \eqref{full_3D} can be further simplified since 
\[ 
\pp{Q}{\bGam} = A(R) \frac{\bGam}{|\bGam|} =Q  \frac{\bGam}{|\bGam|^2} \quad\text{and}\quad   \bGam \times \pp{Q}{\bGam}=0  \, . 
\]
Notice also that for $A(R)= \pi R ^2 $, we have $\mbox{d}A/\mbox{d}R= 2 \pi R$, the circumference of a circle.

\begin{remark}[On further equation simplification] \rm 
We would like to emphasize that any simplification of the model must be done at the level of the Lagrangian $\ell$, for example, by using quadratic elastic energy of certain type,  simplified expressions for moments of inertia, \emph{etc.}  However, once the Lagrangian $\ell$ is chosen, the equations \eqref{full_3D} and their corresponding reductions for two dimensions are completely determined. No further approximations at the level of the equations are possible without losing the exact Lagrangian structure of the system. Such approximations are almost certain to introduce additional, uncontrolled energy and momentum sources and sinks, and we will avoid using them in this paper.  We are in the process of developing a variational integrator for this problem, however, we expect that such a development will be quite intricate. Symplectic and multisymplectic variational integrators for a geometrically exact rod without fluid motion have been only derived recently, see \cite{DeFGBKoRa2014} and \cite{DeGBBLeObRaWe2014}, and the fluid motion presents the substantial difficulty of introducing right-invariant terms in the problem. Multisymplectic discretization for tube conveying incompressible fluid has been considered in \cite{FGBPu2016}.
\end{remark} 

\rem{ 
For the Lagrangian \eqref{Lagrangian_fluid_tube} and with $A$ given by \eqref{modelA}, we have
\begin{align*} 
\frac{\delta \ell}{\delta \boldsymbol{\Omega } }- \mu \frac{\partial Q}{\partial \boldsymbol{\Omega } } &= -KB \boldsymbol{\Omega} | \boldsymbol{\Gamma} |- \mathbb{J}  \boldsymbol{\Omega} | \boldsymbol{\Gamma} |, \quad B:= \frac{1}{2} \rho | \boldsymbol{\gamma} + \boldsymbol{\Gamma} u| ^2 - \mu \\
\frac{\delta \ell}{\delta \boldsymbol{\Gamma} }- \mu \frac{\partial Q}{\partial \boldsymbol{\Gamma} } &=  \rho Au( \boldsymbol{\gamma} + \boldsymbol{\Gamma} u)| \boldsymbol{\Gamma} | - \lambda ( \boldsymbol{\Gamma} - \boldsymbol{\chi} )| \boldsymbol{\Gamma} | + (f- A \mu ) \frac{\boldsymbol{\Gamma} }{| \boldsymbol{\Gamma} |} \\
\dede{\ell}{u}u - \mu Q &=Q( mu- \mu ),\quad m:= \frac{1}{Q}\frac{\delta \ell}{\delta u}= \rho \bGam \cdot \big( \bgam + \bGam u \big),
\end{align*}
where $f$ denotes the integrand function of the Lagrangian, namely $\ell ( \omega  , \boldsymbol{\gamma} ,\Omega , \boldsymbol{\Gamma} ,u)= \int_0^L f( \omega  , \boldsymbol{\gamma} ,\Omega , \boldsymbol{\Gamma} ,u) |\bGam| \mbox{d} s$. 
}

\subsection{Conservation laws for gas motion and Rankine-Hugoniot conditions} 
Since we are concerned with the flow of compressible fluids, it is natural to ask about the existence of shock waves and the conditions the shock solutions must satisfy at the discontinuity. In the one-dimensional motion of a compressible fluid, the constraints on jumps of quantities across the shock are known as the Rankine-Hugoniot conditions. Let us for shortness denote by $[a]$ the jump of the quantity $a$ across the shock, and $c$ the velocity of the shock. The classical Rankine-Hugoniot conditions for the one-dimensional motion of a compressible fluid 
gives the continuity of the corresponding quantities as
\begin{align} 
c [\rho] &= [ \rho u ] \quad \mbox{(mass)} \, , 
\label{RH_mass_0}
\\ 
c [\rho u ] & = [ \rho u^2 + p ] \quad \mbox{(momentum)} \, , 
\label{RH_momentum_0}
\\ 
c[E] & = \left[ \left( \frac{1}{2} \rho u^2 +  \rho e + p \right) u \right] \, ,   \quad E=\frac{1}{2} \rho u^2 + \rho e 
 \quad \mbox{(energy)} \, ,
\label{RH_energy_0} 
\end{align} 
see, e.g., \cite{Wi1974}, where we have defined $E$ to be the total energy density of the gas. It is useful and, in our opinion, rather non-trivial, to derive the corresponding conditions for the gas moving in a tube that can freely deform in the three dimensional space and with expandable cross-section. As far as we know, no such conditions have been derived before, and the derivation of these equations shows the full prowess of the geometric methods. The main difficulty comes in the derivation of the analogue of the energy equation \eqref{RH_energy_0}, as the analogues of equations \eqref{RH_mass_0} and \eqref{RH_momentum_0} are rather straightforward. 

\rem{ 
To get the analogue of \eqref{RH_momentum_0}, we use \eqref{meq} and the equalities
\[
-\frac{1}{\rho} \partial_s p = - \partial _s \left( \rho \frac{\partial e}{\partial \rho } \right) - \frac{\partial e}{\partial \rho } \partial _s\rho = - \partial _s \left(\rho\frac{\partial e}{\partial \rho }+ e\right) +\frac{\partial e}{\partial S}\partial _s S=- \partial _s \left( \frac{p}{\rho } + e\right) +\frac{\partial e}{\partial S}\partial _s S.
\]
For a constant entropy, \ldots 
} 

The mass conservation \eqref{xi_cons} is already written in a conservation law form. 
We rewrite the balance of fluid momentum in the following form
\begin{equation} 
\partial _t \big( \xi \bGam \cdot \left( \bgam+ u \bGam \right) \big)+\partial _s \big( u\xi \bGam \cdot \left( \bgam+ u \bGam \right) +pQ\big)
 -\xi ( \boldsymbol{\gamma} +u\boldsymbol{\Gamma} )\! \cdot \!( \partial _s\boldsymbol{\gamma} + u \partial _s \boldsymbol{\Gamma} )= p \partial _s Q.
\label{gas_momentum_cons} 
\end{equation} 
Note that we cannot use \eqref{meq} to compute the Rankine-Hugoniot conditions, even though the form of the equations is similar. This is due to physical requirement that the conservation laws across the shock must include the conservation of mass, fluid momentum and energy. While many other conservation laws are formally possible, only these three conservation laws make physical sense, see \cite{Wi1974}.

The derivation of the corresponding energy equation is rather tedious and we will only sketch it, presenting the final solution. We define the total energy $\mathbb{E}$, including the thermal and mechanical terms, and the energy density $E$ as 
\begin{equation}
\mathbb{E}=\int_0^L \!E \mbox{d} s \, , \qquad E:= \xi e + \dot R \pp{\ell_0}{\dot R} + \bom \cdot \pp{\ell_0}{\bom} + \bgam \cdot \pp{\ell_0}{\bgam} 
+ u \pp{\ell_0}{u} - \ell_0 \, .
\label{E_tot_def}
\end{equation} 
Then, performing appropriate substitution for time derivatives of the terms in \eqref{E_tot_def} using equations of motion \eqref{full_3D}, we obtain the conservation laws for the energy density $E$ as 
\[
\partial_t E+ \partial_s J =0
\]
for the energy flux $J$ given by
\begin{equation} 
\quad J  := \bom \cdot \pp{\ell_0}{\bOm} + \bgam \cdot \pp{\ell_0}{\bGam} + \dot R \pp{\ell_0}{R'} + u^2 \pp{\ell_0}{u} - \xi u \pp{\ell_0}{\xi} + p \bgam \cdot \pp{Q}{\bGam} + \left( \frac{p}{\rho}+e \right) \xi u .
\label{cons_energy} 
\end{equation} 
Notice an interesting symmetry between time derivatives and spatial derivatives in the expression for the energy flux $J$ in \eqref{cons_energy}. 
Taking only the jumps at the discontinuous terms, we arrive to the following conservation laws for the shock wave moving with velocity $c$: 
\begin{align} 
& c [\rho Q ] = [ \rho Q u ] 
\label{RH_mass}
\\ 
&c \left[ \xi \bGam \cdot \left( \bgam+ u \bGam \right) \right] 
= \left[  \xi u \bGam \cdot (\bgam + u \bGam ) + Q p  \right] 
\label{RH_momentum} 
\\ 
& c \left[ \xi \left( e + \frac{1}{2} \left| \bgam + \bGam u \right|^2 \right) \right] = 
\left[ \frac{1}{2} \xi u \left| \bgam + \bGam u \right|^2+ \frac{p \xi}{\rho |\bGam|^2 } \bGam \cdot (\bgam+ \bGam u) + \xi u e    \right] .
\label{RH_energy} 
\end{align} 
We can further simplify these equations by noticing that at the jump, the continuity of the tube is preserved (or even ${\mathcal C}^1$) so $\bgam$, $\bGam$ and $Q$ are continuous. Therefore, remembering that $\xi=\rho Q$, Rankine-Hugoniot conditions above simplify to 
\begin{align} 
& c [\rho ] = [ \rho u ] 
\label{RH_mass}
\\ 
&c \left[ \rho \right]  \bGam \cdot  \bgam+ c\left[ \rho u \right] | \bGam |^2  
= \left[  \rho u \right] \bGam \cdot \bgam + \left[ \rho u^2 \right] | \bGam |^2  + [ p]  
\label{RH_momentum} 
\\ 
& c \left[ \rho \left( e + \frac{1}{2} \left| \bgam + \bGam u \right|^2 \right) \right] = 
\left[ \frac{1}{2} \rho u \left| \bgam + \bGam u \right|^2+ \frac{p }{ |\bGam|^2 } \bGam \cdot (\bgam+ \bGam u) + \rho u e    \right] .
\label{RH_energy} 
\end{align}

For comparison with the classic Rankine-Hugoniot condition we set the tube to be circular, so $R'=0$ and $\dot R=0$, and not moving and straight, so $\bgam=\mathbf{0}$, $\bom=\mathbf{0}$, $\bOm=\mathbf{0}$ $\bGam=\mathbf{E}_1$, hence and $Q=A_0$. Then, the mass conservation law \eqref{RH_mass} reduces to \eqref{RH_mass_0}, \eqref{RH_momentum} reduces to \eqref{RH_momentum_0} and  \eqref{RH_energy} to \eqref{RH_energy_0}. We note that the extensions of the Rankine-Hugoniot conditions we have derived here are valid for all configurations of the tube in our framework, and they account for motion of the fluid, tube's motion in space and its deformations, and also the expansion/contraction of tube's cross-section coming from the dynamics of the tube's radius $R(t,s)$. 

\subsection{On the Rankine-Hugoniot conditions for traveling waves} 
{\rm 
Later in this paper, we will study the nonlinear traveling waves, which are solutions that depend on $s$ and $t$ through the combination $x=s-ct$, where $c$ is a constant determined by the dynamics. 
When $\Lambda(t,s)=\Lambda(s-ct)$ and $\mathbf{r}(t,s)=\mathbf{r} (s-ct)$, then by definition we have 
$\bom  = - c \bOm $ and $\bgam =-c (\bGam-\mathbf{E}_1) $. Equation \eqref{RH_mass} is unchanged, while \eqref{RH_momentum} and \eqref{RH_energy}  simplify to give  
\begin{equation} 
\begin{aligned} 
& \left[ \rho (u-c)^2  | \bGam |^2  + \rho(u-c) c \bGam \cdot \mathbf{E}_1 + p \right]  =0 
\\ 
& \big[  \rho e  (u-c)  + \frac{1}{2}  \rho (u-c) \left( c^2 + 2 (u-c) c \bGam \cdot \mathbf{E}_1 + (u-c)^2 |\bGam|^2  \right)   \\ 
& \qquad \qquad \qquad \qquad +
\frac{p}{|\bGam|^2} \left( |\bGam|^2 (u-c) + c \bGam \cdot \mathbf{E}_1  \right) \big] = 0 .
\label{RH_moving} 
\end{aligned} 
\end{equation} 
Taking $\bGam=\mathbf{E}_1$ recovers the standard Rankine-Hugoniot conditions for momentum and energy balance across the shock \eqref{RH_momentum_0} and \eqref{RH_energy_0}.

For a perfect gas, we have $p= \frac{R}{C_v} \rho e = (\gamma-1) \rho e$, where $\gamma=C_p/C_v$, therefore the pressure can be eliminated from  equations \eqref{RH_moving}  leading to 
\begin{equation} 
\begin{aligned} 
& \left[ \rho (u-c)^2  | \bGam |^2  + \rho(u-c) c \bGam \cdot \mathbf{E}_1 +  (\gamma-1) \rho e \right]  =0 
\\ 
& \big[  \rho e  (u-c)  + \frac{1}{2}  \rho (u-c) \left( c^2 + 2 (u-c) c \bGam \cdot \mathbf{E}_1 + (u-c)^2 |\bGam|^2  \right)\\ 
& \qquad \qquad \qquad \qquad +
(\gamma-1) \frac{\rho e}{|\bGam|^2} \left( |\bGam|^2 (u-c) + c \bGam \cdot \mathbf{E}_1  \right) \big] = 0 .
\label{RH_moving_perfect} 
\end{aligned} 
\end{equation} 
 We shall use  equations \eqref{RH_moving} later in \S\ref{sec:traveling_waves} for computations of shock waves in the fluid part for traveling waves. 
}

\subsection{Incompressible fluids}\label{sec_incomp}

The incompressibility of the fluid motion is imposed by requiring that the mass density per unit volume is a constant number:
\[
\rho (t,s)=\rho _0
\]
Given the expression \eqref{Q_function} of the area in terms of the tube's variables, the relation \eqref{relation_xi_rho_Q}  still holds with $ \rho =\rho _0$. It thus induces a holonomic constraint in the configuration space $\mathcal{Q} $, namely
\begin{equation}\label{NH_constraint} 
Q( \boldsymbol{\Omega} , \boldsymbol{\Gamma} , R)=\frac{\xi}{\rho_0 }.
\end{equation} 
Recall that $ \xi = (\xi _0 \circ \varphi^{-1} ) \partial _s \varphi^{-1} $, so \eqref{NH_constraint} can be written as
\[
Q( \boldsymbol{\Omega} , \boldsymbol{\Gamma} , R) = (Q _0 \circ \varphi^{-1} ) \partial _s \varphi^{-1},
\] 
where $Q_0 = {\xi _0}/{ \rho _0}$. This is the form of the holonomic constraint that is imposed in \cite{FGBPu2014,FGBPu2015}, leading to the Lagrange multiplier $ \mu $. In \cite{FGBPu2014,FGBPu2015} it was however assumed that the shape of the cross-section, and hence the variable $Q$, is uniquely determined by the variables $\bOm$ and $\bGam$, corresponding to $R$ being constant in our framework.
Note that this holonomic constraint implies that $ Q$, similarly with $ \xi =\rho _0Q$ is advected by the fluid motion, namely, it verifies
\[
\partial _t Q+\partial _s (Qu)=0,
\]
which is not true in the compressible case.
With a Lagrangian of the form 
\[
\ell( \boldsymbol{\omega} ,\boldsymbol{\gamma} , \boldsymbol{\Omega} ,  \boldsymbol{\Gamma} ,u, R,\dot{R})= 
\int_0^L \ell_0( \boldsymbol{\omega} ,\boldsymbol{\gamma} , \boldsymbol{\Omega} ,  \boldsymbol{\Gamma} ,u ,R,\dot{R}, R')\, \mbox{d} s \, ,
\]
the variational principle
\[
\delta \int_0^T \left[ \ell( \boldsymbol{\omega} , \boldsymbol{\gamma} ,\boldsymbol{\Omega} , \boldsymbol{\Gamma} , u, R, \dot R )+ \int_0^L \mu  \Big( Q(\boldsymbol{\Omega}, \boldsymbol{\Gamma} , R)- (Q_0 \circ \varphi ^{-1} ) \partial _s \varphi^{-1} \Big) \mbox{d}s\right]\mbox{d}t=0
\]
with respect to variations  \eqref{delta1}, \eqref{delta2}, 
$\delta \varphi$, and $\delta R$, yields the system
\begin{equation}\label{full_3D_incompressible} 
\left\lbrace\begin{array}{l}
\displaystyle\vspace{0.2cm}\frac{D}{Dt}\pp{\ell_0}{\bom}+\bgam\times\pp{\ell_0}{\bgam} +\frac{D}{Ds}\left(  \pp{\ell_0}{\bOm} +\mu  \pp{Q}{\bOm} \right)+\bGam\times 
 \left( 
\pp{\ell_0}{\bGam} + \mu  \pp{Q}{\bGam} 
\right) =0
\\
\displaystyle\vspace{0.2cm}\frac{D}{Dt}\pp{\ell_0}{\bgam} + \frac{D}{Ds}\left( \pp{\ell_0}{\bGam} +  \mu  \pp{Q}{\bGam}   \right)=0\\
\displaystyle\vspace{0.2cm}\prt_t\frac{ \partial  \ell_0}{ \partial  u}  + u \partial _s\frac{ \partial  \ell_0}{ \partial  u}+ 2  \frac{ \partial  \ell_0}{ \partial  u} \partial _s u =  -Q \partial_s\mu  \\
\displaystyle\vspace{0.2cm} 
\partial_t \pp{\ell_0}{\dot R} - \partial^2_s \pp{\ell_0}{ R''} +  \partial_s \pp{\ell_0}{ R'} -  \pp{\ell_0}{ R}-\mu \frac{\partial Q}{\partial R}=0 
\\
\displaystyle\vspace{0.2cm} \partial _t \boldsymbol{\Omega} = \boldsymbol{\omega} \times \boldsymbol{\Omega} +\partial _s  \boldsymbol{\omega}, \qquad  \partial _t \boldsymbol{\Gamma} + \boldsymbol{\omega} \times \boldsymbol{\Gamma} = \partial _s \boldsymbol{\gamma} + \boldsymbol{\Omega} \times \boldsymbol{\gamma}\\
\displaystyle \partial _t Q + \partial _s (Q u)=0\,.
\end{array}\right.
\end{equation}
A direct comparison with the compressible system \eqref{full_3D} shows that the Lagrange multiplier $ \mu $ plays the role of the pressure, in complete analogy with pressure definition for  incompressible fluid mechanics.   When $R$ is assumed to be constant, this model recovers the one derived in \cite{FGBPu2014,FGBPu2015}.

\rem{
\todo{FGB: In the incompressible case we have neglected the internal energy $e$. If we take it into account, we have the same equations \eqref{full_3D_incompressible} with $ \mu $ replaced by $\mu +p$, $p= \rho ^2 \frac{\partial e}{\partial\rho }$, the sum of a geometric and a thermodynamic pressure.\\
Note that I have changed the convention $ \mu \rightarrow - \mu $ with respect to (2.17) in our paper \cite{FGBPu2015}, to have $\mu $ analogous to $p$ and not $-p$.}
}

\section{Geometric approach to the variational principle}\label{Geometric_approach} 

In this section, we justify the variational principle used earlier by showing that it is obtained from a reduction by symmetry of the classical Hamilton principle of Lagrangian mechanics. We first consider an abstract setting that involves both a left and a right symmetry of the Lagrangian, which correspond to the elastic and fluid components of the system, written on the infinite dimensional configuration space $\mathcal{Q}$. This general setting rigorously explains the link between the symmetries of the problem, the form of the variational principle and the choice of the variables.
We shall apply this setting to the tube conveying fluid, both in the compressible and incompressible cases.
We shall then apply this setting to the tube conveying fluid, both in the compressible and incompressible cases, thereby justifying the approach developed in Section \ref{sec:general}.

\paragraph{Lagrangian and symmetries.} Consider two Lie groups $G$ and $H$, with Lie algebras $ \mathfrak{g}  $ and $ \mathfrak{h}  $. We assume that $G$ acts on the left on a manifold $P$ and that $H$ acts on the right on a manifold $N$. We shall denote these actions as
\begin{equation}\label{LR_actions}
\begin{aligned} 
&\Phi : G \times P \rightarrow P, \quad (g,p) \mapsto \Phi _g(p), \quad \Phi_g\circ \Phi_h=\Phi_{gh}\\
&\Psi  : H \times N \rightarrow N, \quad (h,n) \mapsto \Psi _h(n), \quad \Psi_g\circ \Psi_h= \Psi_{hg}.
\end{aligned}
\end{equation}  
Given the Lie algebra elements $ \zeta \in \mathfrak{g}  $ and $u\in\mathfrak{h}  $, we consider the associated infinitesimal generators $ \zeta _P$ and $ u_N$, which are the vector fields on $P$, resp., $N$ defined by
\begin{equation}\label{infinit_gen}
\zeta _P(p):= \left.\frac{d}{d\varepsilon}\right|_{\varepsilon=0} \Phi _{ \exp( \varepsilon \zeta )}(p), \quad\text{resp.,}\quad  u _N(n):= \left.\frac{d}{d\varepsilon}\right|_{\varepsilon=0} \Psi _{ \exp( \varepsilon u )}(n),
\end{equation}  
where $\exp$ denotes the exponential map of the Lie groups. Any action of a Lie group $G$ on a manifold $P$, induces an action of $G$ on the cotangent bundle, or phase space, $T^*P$ of $P$. This cotangent lifted action preserves the canonical symplectic form on $T^*P$.

As we shall see below, the equations of motion naturally involve the expression of the momentum maps associated to the cotangent lifted actions of $G$ and $H$ on $T^*P$ and $T^*N$. These momentum maps, sometimes referred to as simply the  cotangent lift momentum maps, are given by
\begin{equation}
\begin{aligned} 
&\mathbb{J} _L: T^*P \rightarrow  \mathfrak{g}  ^\ast , \quad \left\langle  \mathbb{J}  _L( \alpha _p), \zeta \right\rangle = \left\langle \alpha _p, \zeta _P(p) \right\rangle\\
&\mathbb{J} _R: T^*N \rightarrow  \mathfrak{h}  ^\ast , \quad \left\langle  \mathbb{J}  _R( \alpha _n), u \right\rangle = \left\langle \alpha _n, u _N(n) \right\rangle ,
\end{aligned} 
\label{MomMap_def}
\end{equation} 
where  $ \mathfrak{g}  ^\ast $, $ \mathfrak{h}  ^\ast $ are the dual spaces to $ \mathfrak{g}  $, $ \mathfrak{h}  $, and where $\alpha _p \in T^*_pP$, $ \alpha _n \in T^*_nN$, with $T^*P$, $T^*N$ the cotangent bundles of $P$ and $N$. In \eqref{MomMap_def}, the pairing on the left hand side of the equality is between the Lie algebra and its dual, while the pairing on the right hand side of the equality is between the tangent space and the cotangent space at a given point of the manifold.
Note that $\alpha _p \in T^*_pP$ (resp.,  $ \alpha _n \in T^*_nN$), being the element of the cotangent bundle, contains the information of both the vectors from the cotangent spaces and the base points $p$ (resp., $n$) themselves. Therefore, the momentum map $\mathbb{J}_L$ (resp.,  $\mathbb{J}_R$) is a function of a vector from the cotangent space \textit{and} the base point. For instance, if the manifold $P$ is a vector space, we have $T^*P=P\times P^*$ and an element in the cotangent bundle is a couple $\alpha_p=(p,\alpha)$. We refer the reader to, e.g., \cite{MaRa2002} for the definition of momentum maps and their properties.

\medskip

In general, given the configuration manifold $Q$ of a mechanical systems and its Lagrangian $\mathsf{L}:TQ\rightarrow \mathbb{R}$, the Hamilton principle reads
\begin{equation}\label{HP}
\delta\int_0^T \mathsf{L}\big(q(t),\dot q(t)\big) \mbox{d}t=0,
\end{equation}
for arbitrary variations $\delta q(t)$ of the curve $q(t)$ with fixed endpoints. In \eqref{HP}, and in the whole paper, we use the local notation $(q,\dot q)$ for the arguments of the Lagrangian, however, our treatment is intrinsic and valid on arbitrary configuration manifolds $Q$.

We now suppose that the configuration manifold is $Q=G \times H$ and consider the Lagrangian function $\mathsf{L}:T(G \times H) \rightarrow \mathbb{R}$. The Hamilton principle reads
\begin{equation}\label{HP_G} 
\delta \int_0^T \mathsf{L}\big(g(t), \dot g(t), h(t), \dot h(t)\big) \mbox{d}t=0,
\end{equation} 
for arbitrary variations $ \delta g(t)$ and $ \delta h(t)$ vanishing at $t=0, T$.

Let us suppose that $\mathsf{L}$ depends parametrically on some reference values $p_0 \in P$, $n_0\in N$ and assume that $\mathsf{L}$ is invariant under the action of the subgroup $G_{p_0} \!\times \!H_{n_0} \subset G\times H$, where $G_{p_0}=\{ g \in G\mid \Phi _g(p_0)=p_0\}$ and $H_{n_0}= \{ h \in H \mid \Psi _h(n_0)=n_0\}$ are the isotropy subgroups of $p_0$ and $n_0$, with respect to the actions given in \eqref{LR_actions}. This invariance is written as
\begin{equation}\label{symmetry}
\mathsf{L}\big(  \bar g g, \bar g\dot g, h \bar h, \dot h \bar h\big) = \mathsf{L}\big(g, \dot g, h, \dot h\big), \quad \text{for all} \quad (\bar g, \bar h ) \in G_{p_0}\!\times\! H_{n_0},
\end{equation} 
where we note that $G_{p_0}$ acts on the right on $G$ and $H_{n_0}$ acts on the left on $H$.

\paragraph{Lagrangian reduction.} We now use the symmetry \eqref{symmetry} to rewrite the equations of motion on the reduced space by following the process of Lagrangian reduction by a subgroup of the Lie group configuration space, see \cite{HoMaRa1998}, \cite{GBRa2009}, and \cite{GBTr2010} for linear, affine, and general actions, respectively, and various applications. We shall follow the setting of \cite{GBTr2010}.
From the invariance \eqref{symmetry}, $\mathsf{L}$ induces a reduced Lagrangian $\ell$ on the quotient manifold
\[
\big(T(G \times H)\big)/(G_{p_0}\! \times \!H_{n_0})
\]
consisting of the tangent bundle of the configuration space, divided by the symmetry group. We identify this quotient with the manifold $ \mathfrak{g}  \times \mathfrak{h} \times \mathcal{O}$, where $ \mathcal{O} $ is the $(G \times H)$-orbit of $(p_0, n_0) \in P \times N$, as follows
\begin{equation}\label{identification_Lagr}
\begin{aligned} 
\big(T(G \times H)\big)/(G_{p_0}\! \times\! H_{n_0}) &\longrightarrow \mathfrak{g}  \times \mathfrak{h} \times \mathcal{O}\\
[g, \dot g, h, \dot h] &\longmapsto \big( g ^{-1} \dot g, \dot h h ^{-1} , \Phi _{g^{-1}} (p_0) , \Psi _{h^{-1} }(n_0) 
\big),
\end{aligned}
\end{equation} 
where $[g, \dot g, h, \dot h]$ denotes the equivalence class of $(g, \dot g, h, \dot h)$ in the quotient manifold $\big(T(G \times H)\big)/(G_{p_0}\! \times\! H_{n_0})$.
Consistently with this identification, the reduced curves associated to $\big( g(t), h(t)\big) \in G  \times H$ are
\begin{equation}\label{reduced_curves} 
\begin{aligned} 
\mbox{Left-invariant}: \quad \zeta (t)&=g (t)^{-1} \dot g(t) \in \mathfrak{g}   \qquad p(t)=\Phi _{g(t)^{-1}} (p_0)\in P\, ,   \\
\mbox{Right-invariant}:  \quad u(t)&=\dot h(t) h(t) ^{-1}  \in\mathfrak{h}  \qquad \, n(t)=\Psi _{h(t)^{-1} }(n_0)\in N.
\end{aligned}
\end{equation} 
The Hamilton principle \eqref{HP} induces the following variational principle for the reduced Lagrangian $\ell:  \mathfrak{g}  \times \mathfrak{h} \times \mathcal{O} \rightarrow \mathbb{R}  $
\begin{equation}\label{red_HP} 
\delta \int_0^T  \ell\big( \zeta (t), u(t), p(t),n(t)\big) \mbox{d} t=0,
\end{equation} 
for variations $ \delta \zeta (t) $, $ \delta u (t) $, $ \delta p (t) $, $ \delta n(t) $ given by
\begin{equation}\label{red_variations}
\begin{aligned} 
\delta \zeta &= \dot { \sigma }+[ \zeta , \sigma ] \qquad  \delta p= - \sigma _P(p) \\
\delta u &= \dot v-[u,v] \qquad \, \,\delta n= - v_N(n),
\end{aligned}
\end{equation}
where $ \sigma (t) \in \mathfrak{g}  $ and $ v(t) \in \mathfrak{h}  $ are arbitrary curves vanishing at $t=0,T$.
The expressions \eqref{red_variations} are obtained by computing the variations of the curves $ \zeta (t)= g(t) ^{-1} \dot g(t)$, $ u(t)= \dot h(t) h(t) ^{-1} $, $ p(t)= \Phi _{ g(t)^{-1} }(p_0)$, $n(t) =  \Psi _{ h(t)^{-1} }(n_0)$ induced by the variations $ \delta g (t) $, $ \delta h(t) $ of the curves $g(t)$ and $h(t)$, where we defined $ \sigma := g ^{-1} \delta g\in \mathfrak{g}  $ and $ v:= \delta h h ^{-1}\in\mathfrak{h}  $.

By applying the variational principle \eqref{red_HP}-\eqref{red_variations}, we get the system of equations
\begin{equation}\label{general_system} 
\left\{ 
\begin{array}{l}
\displaystyle\vspace{0.2cm}\frac{d}{dt} \frac{\delta \ell}{\delta  \zeta }- \operatorname{ad}^*_ \zeta \frac{\delta \ell}{\delta \zeta }+ \mathbb{J}  _L \left( \frac{\delta \ell}{\delta p} \right) =0\\
\displaystyle\frac{d}{dt} \frac{\delta \ell}{\delta u}+\operatorname{ad}^*_ u \frac{\delta \ell}{\delta u}  + \mathbb{J}  _R\left( \frac{\delta \ell}{\delta n} \right) =0,
\end{array}
\right.
\end{equation}
see \cite{GBTr2010} for details.
The partial derivatives $\frac{\delta \ell}{\delta \zeta } \in \mathfrak{g}  ^\ast $ and  $ \frac{\delta \ell}{\delta p}  \in T^*_pP$ are defined as
\[
\left\langle \frac{\delta \ell}{\delta \zeta }, \delta \zeta \right\rangle = \left.\frac{d}{d\varepsilon}\right|_{\varepsilon=0} \ell\big( \zeta + \varepsilon \delta \zeta ,u,p,n\big) \quad\text{and}\quad \left\langle\frac{\delta \ell}{\delta p}, \delta p \right\rangle = \left.\frac{d}{d\varepsilon}\right|_{\varepsilon=0} \ell\big( \zeta , u, p(\varepsilon) , n\big), 
\]
where $p( \varepsilon) \in P$ is a curve with $p(0)=p$ and $ \left.\frac{d}{d\varepsilon}\right|_{\varepsilon=0} p( \varepsilon )= \delta p \in T_pP$. 

The coadjoint operator $ \operatorname{ad}^*_ \zeta : \mathfrak{g}  ^\ast \rightarrow \mathfrak{g} ^\ast $ in the first equation in \eqref{general_system} is defined by $ \left\langle  \operatorname{ad}^*_ \zeta \mu , \sigma\right\rangle = \left\langle \mu , [\zeta , \sigma ]\right\rangle $, for $ \mu\in\mathfrak{g}  ^\ast $, $ \zeta , \sigma\in\mathfrak{g}  $, where $[ \zeta , \sigma ]$ is the Lie bracket on $ \mathfrak{g}  $. Similar definitions hold for $ \frac{\delta \ell}{\delta u}\in \mathfrak{h}  ^\ast$, $ \frac{\delta \ell}{\delta n}  \in T^*_nN$, and $ \operatorname{ad}^*_ u: \mathfrak{h} ^\ast \rightarrow\mathfrak{h} ^\ast $. 

The system \eqref{general_system} is accompanied with the equations for the curves $p(t)$ and $n(t)$, which follow from their definition in \eqref{reduced_curves}, namely
\begin{equation}\label{m_n_advection} 
\dot p+ \zeta _P(p)=0, \quad \dot n+ u_N(n)=0,
\end{equation}
where $ \zeta _P$ and $ u_N$ are the infinitesimal generators defined in \eqref{infinit_gen}.

\medskip

The considerations of this paragraph are summarized in the following theorem, which is a special instance of a result in \cite{GBTr2010}, see also \cite{GBRa2009}, \cite{HoMaRa1998}.

\begin{theorem}\label{theorem} Consider a Lagrangian $\mathsf{L}:T(G\times H)\rightarrow \mathbb{R}$ with the invariance described in \eqref{symmetry} and define the associated reduced Lagrangian $\ell:\mathfrak{g}\times\mathfrak{h}\times\mathcal{O}\rightarrow\mathbb{R}$. Consider a curve $(g(t), h(t))\in G\times H$ and define $(\zeta(t),u(t), p(t),n(t))\in \mathfrak{g}\times\mathfrak{h}\times\mathcal{O}$ as earlier. In particular, equations \eqref{m_n_advection} hold. Then the following are equivalent.
\begin{itemize}
\item[{\bf (i)}] The curve $(g(t), h(t))$ is critical for the Hamilton principle associated to $\mathsf{L}$;
\item[{\bf (ii)}] The curve $(g(t), h(t))$ is a solution of the Euler-Lagrange equations associated to $\mathsf{L}$;
\item[{\bf (ii)}] The curve $(\zeta(t),u(t), p(t),n(t))$ is critical for the reduced Hamilton principle \eqref{red_HP}--\eqref{red_variations} associated to $\ell$;
\item[{\bf (iv)}] The curve $(\zeta(t),u(t), p(t),n(t))$ is a solution of the reduced Euler-Lagrange equations \eqref{general_system}.
\end{itemize}
\end{theorem}

\paragraph{Application to the expandable tube conveying compressible flows.} In this application, $g(t)$ is given by the tube variables $(\Lambda(t), \mathbf{r}(t))$, and $h(t)$ is given by the flow variable $\varphi(t)$. In addition to $\Lambda(t)$, $\mathbf{r}(t)$, and $\varphi(t)$, the geometrically exact expandable tube also involves the variable $R(t)$ but this variable is not involved in the reduction process. We shall briefly mention the inclusion of such variables later. Also, in this application, the configuration space associated to $\varphi(t)$ is not a Lie group
however for simplicity we restrict the discussion in this paragraph to the case of a Lie group and leave the more general case as an exercise for the reader as the resulting formulas are similar.
The Lie algebra element $\zeta$, resp., $u$ in the general setting corresponds to the convective velocities $(\boldsymbol{\omega}, \boldsymbol{\gamma})$, resp., the Eulerian velocity $u$ in the application.
The abstract expressions $\zeta= g^{-1}\dot g$ and $u=\dot h h^{-1}$ for the convective and Eulerian velocities recover the expressions \eqref{Eulerian_velocity_u} and the first two expressions in \eqref{def_conv_var}.

\medskip

In this application, the manifold $P$ contains two variables, the convective angular and linear deformation gradients $\boldsymbol{\Omega} $ and $\boldsymbol{\Gamma}$ on which a group element $(\Lambda,\mathbf{r})$ acts \textit{affinely} on the \text{left} as
\begin{equation}\label{left_action}
( \boldsymbol{\Omega} , \boldsymbol{\Gamma}) \mapsto \operatorname{Ad}_{( \Lambda , \mathbf{r} )}(  \boldsymbol{\Omega} , \boldsymbol{\Gamma}) + ( \Lambda , \mathbf{r} ) \partial _s ( \Lambda , \mathbf{r} ) ^{-1},
\end{equation}
where $\operatorname{Ad}_{( \Lambda , \mathbf{r} )}$ denotes the adjoint action of $SE(3)$. We refer to \cite{GBRa2009} for more details regarding affine actions of the type \eqref{left_action} and their applications in the context of Lagrangian reduction by symmetry.
The manifold $N$ contains two variables, the mass density per unit length $\xi$ and the specific entropy $S$, on which $\varphi$ acts \textit{linearly} on the \textit{right} as
\begin{equation}\label{right_action}
\xi  \mapsto ( \xi \circ \varphi ) \partial _s\varphi \quad\text{and}\quad S \mapsto S \circ \varphi .
\end{equation}
The actions \eqref{left_action} and \eqref{right_action} correspond to the actions \eqref{LR_actions} in the general setting explained above.
The reference value $p_0$ is given by the initial values for the mass density and specific entropy, \emph{i.e.}, $p_0=(\xi_0,S_0)$. The reference value $n_0$ is chosen as $n_0= (\boldsymbol{\Omega}_0,\boldsymbol{\Gamma}_0)=(\mathbf{0}, \mathbf{0})$.

\medskip
To complete the variational considerations of the tube with expandable wall conveying compressible fluid, we need to remember that there are additional non-reduced variables $a(t)\in K$, in some configuration manifold $K$, \emph{e.g.}, $a=R$, which we did not consider in this abstract treatment, as we were focusing on the symmetry reduction only. These variables will satisfy the Euler-Lagrange equations in the abstract form
\begin{equation} 
\label{EL_a} 
\pp{}{t} \dede{\ell}{\dot a} - \dede{\ell}{a}=0 \, . 
\end{equation} 
For our considerations in this paper, we set $a(t,s)=R(t,s)$, a scalar function. Theorem \ref{theorem} easily extends to this case, the unreduced and reduced space being given by $T(G\times H \times K)$ and $\mathfrak{g}\times\mathfrak{h}\times\mathcal{O}\times TK$, respectively.

\medskip

With these choices, the variational principle \eqref{red_HP}--\eqref{red_variations} recovers the variational principle \eqref{min_action_gas}--\eqref{delta_S}. From Theorem \ref{theorem}, it follows that the variational principle \eqref{min_action_gas}--\eqref{delta_S} is a symmetry reduced version of the classical Hamilton principle \eqref{HP_total} on the configuration manifold $\mathcal{Q}$ of the mechanical system. This rigorously justifies from first principles the variational principle  \eqref{min_action_gas}--\eqref{delta_S} used to derive our model of geometrically exact expandable tube conveying gas.

\medskip

We shall remark that a more complex shape model may require the variable $a$ to be \emph{e.g.}, a two-dimensional function, for example, describing the major semi-axes of an ellipse as a cross-section, or be a multi-variable function describing shape functions of increasing complexity. The abstract equation \eqref{EL_a} will still be valid, although it will need to be re-written in an explicit form for further analysis.

\paragraph{Incompressibility of the  fluid as a holonomic constraint.} Still in the abstract setting described above, we consider a certain quantity of interest $\rho \in V$ which depends on the variables $p$ and $n$, i.e., we assume
\[
\rho=F(p,n),
\]
for a given function $ F :P \times N \rightarrow V$, where $V$ is a vector space. In our application, $\rho$ is the mass density per unit volume and the function $F$ is given by the relation \eqref{relation_xi_rho_Q}.

We assume that the quantity $\rho$ is constrained to be a constant in time, i.e., 
\[
F \big(p(t),n(t)\big)= \rho _0= \text{constant},
\]
for all $t$.
In our example this corresponds to the incompressibility constraint in \S\ref{sec_incomp}.
This defines a holonomic constraint which, from the relations on the right hand side of \eqref{reduced_curves}, corresponds to the subset
\[
C:= \left \{(g,h)\mid G \times H \mid F\big( \Phi_{g ^{-1} }(p_0), \Psi_{h^{-1} }(n_0) \big)= \rho _0\right \} \subset G \times H
\]
of the configuration manifold $Q= G \times H$. We assume that $C$ is a submanifold of $Q$.

The associated equations of evolutions are obtained by a standard Lagrange multiplier approach, namely, we replace \eqref{HP_G} by the variational principle
\begin{equation}\label{HP_constraints} 
\delta \int_0^T\left[ \mathsf{L}\big(g(t), \dot g(t), h(t), \dot h(t)\big) + \Big\langle \mu (t),F\big( \Phi_{g(t) ^{-1} }(p_0), \Psi_{h(t)^{-1} }(n_0) \big)-\rho _0 \Big\rangle  \right]  \mbox{d}t=0,
\end{equation}
for variations $ \delta g(t)$ and $ \delta h(t)$, vanishing at $t=0, T$, and variations $ \delta \mu (t)$. The Lagrangian multiplier $ \mu $ is an element of the dual vector space $V ^\ast $.

Assuming the same invariance as before for the Lagrangian, i.e. \eqref{symmetry}, and  observing that the constraint also has the same symmetry, we obtain that \eqref{HP_constraints} induces the following variational principle for the reduced Lagrangian $\ell:  \mathfrak{g}  \times \mathfrak{h} \times \mathcal{O} \rightarrow \mathbb{R}  $:
\begin{equation}\label{red_HP_constraints} 
\delta \int_0^T  \big[ \ell\big( \zeta (t), u(t), p(t),n(t)\big)+ \big\langle \mu (t), F(p(t), n(t))- \rho _0 \big\rangle\big]  \mbox{d} t=0,
\end{equation} 
for variations $ \delta \zeta (t) $, $ \delta u (t) $, $ \delta p (t) $, $ \delta n(t) $ given by \eqref{red_variations} and variations $ \delta \mu (t)$. It results in the system
\begin{equation}\label{general_system_constraint} 
\left\{ 
\begin{array}{l}
\displaystyle\vspace{0.2cm}\frac{d}{dt} \frac{\delta \ell}{\delta  \zeta }- \operatorname{ad}^*_ \zeta \frac{\delta \ell}{\delta \zeta }+ \mathbb{J}  _L \left( \frac{\delta \ell}{\delta p}+ \left[ \frac{\partial F }{\partial p} \right] ^\ast   \mu    \right) =0\\
\displaystyle\vspace{0.2cm}\frac{d}{dt} \frac{\delta \ell}{\delta u}+\operatorname{ad}^*_ u \frac{\delta \ell}{\delta u}  + \mathbb{J}  _R\left( \frac{\delta \ell}{\delta n}+ \left[ \frac{\partial F }{\partial n} \right] ^\ast  \mu   \right) =0\\
F( p,n)=\rho _0.
\end{array}
\right.
\end{equation}
One can immediately deduce the extension of Theorem \ref{theorem} to the case in which a holonomic constraint is imposed, together with the equation \eqref{EL_a} for the non-reduced variables $a$. This setting justifies the variational approach used in \S\ref{sec_incomp} to derive the system \eqref{full_3D_incompressible}.

\paragraph{Hamiltonian structure for the expandable tube conveying compressible flows.} We shall now quickly describe the Hamiltonian side corresponding to the Lagrangian variational formulation described above. We first describe the abstract setting and employ the general process of Poisson reduction symmetry, see \cite{MaRa2002}, to derive the Hamiltonian structure of the geometrically exact expandable tube conveying gas in terms of a noncanonical Poisson bracket.

By assuming that the Lagrangian $\mathsf{L}:T(G \times H) \rightarrow \mathbb{R}  $ is hyperregular, we can define its associated Hamiltonian $\mathsf{H}:T^*( G \times H) \rightarrow \mathbb{R}  $ by the usual Legendre transform. The Euler-Lagrange equations for $\mathsf{L}$ are equivalent to the Hamilton equations for $\mathsf{H}$. These equations are Hamiltonian with respect to the canonical Poisson bracket on $T^*(G\times H)$.

The Hamiltonian $\mathsf{H}$ verifies the same invariance with $\mathsf{L}$ and thus induces a reduced Hamiltonian, denoted $h$, on the quotient manifold $\big(T^*(G \times H)\big)/(G_{p_0}\! \times \!H_{n_0})$. In a similar way with \eqref{identification_Lagr}, we identify this quotient manifold with the manifold $ \mathfrak{g} ^\ast \times\mathfrak{h} ^\ast\times \mathcal{O} $ as
\begin{equation}\label{identification_Ham}
\begin{aligned} 
\big(T^*(G \times H)\big)/(G_{p_0}\! \times\! H_{n_0}) &\longrightarrow \mathfrak{g}^*  \times \mathfrak{h} ^*\times \mathcal{O}\\
[g,  \alpha , h, \beta ] &\longmapsto \big( g ^{-1} \alpha ,  \beta  h^{-1} , \Phi _{g^{-1}} (p_0) , \Psi _{h^{-1} }(n_0) 
\big).
\end{aligned}
\end{equation} 
From \eqref{identification_Ham} we deduce that, on the Hamiltonian side, the reduced curves associated to $(g(t), \alpha (t), h(t), \beta (t)) \in T^*(G \times H)$ are
\begin{equation}\label{reduced_curves_Ham} 
\begin{aligned} 
\mu  (t)&=g (t)^{-1} \alpha (t) \in \mathfrak{g} ^*  \qquad p(t)=\Phi _{g(t)^{-1}} (p_0)\in P\\
\nu(t)&=\beta (t) h(t) ^{-1}  \in\mathfrak{h}^*  \qquad \, n(t)=\Psi _{h(t)^{-1} }(n_0)\in N.
\end{aligned}
\end{equation}
In terms of the reduced Hamiltonian $h( \mu , \nu, p, n)$, the system of equations \eqref{general_system}--\eqref{m_n_advection} can be written in Poisson bracket form as
\[
\dot f=  \{f,h\}_L+\{f,h\}_R,
\]
for the noncanonical Poisson brackets $\{\,,\}_L$ and $\{\,,\}_R$ given by
\begin{equation}\label{PB_R} 
\begin{aligned}
\{f,h\}_L&=- \left\langle \mu ,\left[\frac{\delta f}{\delta \mu },\frac{\delta h}{\delta \mu }\right] \right\rangle + \left\langle\frac{\delta f}{\delta \mu }, \mathbb{J}  _L\left( \frac{\delta h}{\delta p} \right) \right\rangle- \left\langle\frac{\delta h}{\delta \mu }, \mathbb{J}  _L\left( \frac{\delta f}{\delta p} \right) \right\rangle\\
\{f,h\}_R&=+ \left\langle \nu ,\left[\frac{\delta f}{\delta \nu },\frac{\delta h}{\delta \nu }\right] \right\rangle + \left\langle\frac{\delta f}{\delta \nu }, \mathbb{J}  _R\left( \frac{\delta h}{\delta n} \right) \right\rangle- \left\langle\frac{\delta h}{\delta \nu }, \mathbb{J}  _R\left( \frac{\delta f}{\delta n} \right) \right\rangle.
\end{aligned}
\end{equation} 
We refer to \cite{GBTr2010} for the derivation of this type of Poisson bracket by reduction by symmetry of the canonical Poisson bracket on $T^*(G\times H)$.

In terms of the reduced Lagrangian $\ell(\zeta,u,p,n)$ defined on $\mathfrak{g}\times\mathfrak{h}\times\mathcal{O}$, the reduced Hamiltonian on $\mathfrak{g}^*\times\mathfrak{h}^*\times\mathcal{O}$ is obtained from the reduced Legendre transform
\[
h(\mu,\nu,p,n):= \left\langle \mu,\zeta\right\rangle+ \left\langle \nu, u\right\rangle - \ell(\zeta,u,p,n),
\]
for $\mu= \frac{\delta\ell}{\delta \zeta}$ and $\nu= \frac{\delta\ell}{\delta u}$.

 In presence of additional non-reduced variables $a$, as in \eqref{EL_a}, the Poisson bracket \eqref{PB_R} is augmented by a canonical Poisson bracket $\{f,g\}_{\rm can}$ in the variables $(a,p_a)$, with $p_a=\dede{\ell}{\dot a}$. In this case the reduced Hamiltonian and Poisson bracket are defined on $\mathfrak{g}^*\times\mathfrak{h}^*\times\mathcal{O}\times T^*K$, where $K$ is the configuration manifold of the variable $a$.

\medskip

For the geometrically exact expandable tube conveying gas, $\mu=(\boldsymbol{\pi}, \boldsymbol{\mu})$ corresponds to the angular and linear momentum of the tube dynamics, $\nu$ is the fluid momentum, $p=(\boldsymbol{\Omega}, \boldsymbol{\Gamma})$, and $n=(\xi, S)$. The Hamiltonian is defined as
\[
h(\bpi,\bmu,  \bOm,\bGam,\nu,\xi,S)= \int_0^L \big(\bpi \cdot \bom + \bmu \cdot \bgam + \nu  u \big) \mbox{d} s- \ell(\bom,\bgam,\bOm,\bGam,u,\xi,S)
\]
and we have the relations
\begin{equation} 
\label{dual_var_def}
\begin{aligned} 
\bpi & =\dede{\ell}{\bom} \, , \quad \bmu=\dede{\ell}{\bgam} \, , \quad \nu= \dede{\ell}{u}\\
\bom & = \dede{h}{\bpi} \, , \quad \bgam = \dede{h}{\bmu}\, , \quad u=\dede{h}{\nu}.
\end{aligned} 
\end{equation} 
The formula \eqref{PB_R} yields the Poisson brackets
\begin{align*}
\{f,g\}_L
=&
-\int_0^L\bpi\cdot\left(
\dede{f}{\bpi}\times\dede{g}{\bpi}\right)\mbox{d}s
-\int_0^L
\bmu\cdot\left(
 \dede{f}{\bmu}\times\dede{g}{\bpi}
- 
\dede{g}{\bmu}\times \dede{f}{\bpi}
\right)\mbox{d}s
\nonumber\\
&-\int_0^L\bOm\cdot\left(
\dede{f}{\bOm}\times\dede{g}{\bpi}
-
\dede{g}{\bOm}\times\dede{f}{\bpi}
\right)\mbox{d}s+ \int_0^L \left( 
\dede{f}{\bOm}\cdot\partial_s\dede{g}{\bpi}
- 
\dede{g}{\bOm}\cdot\partial_s\dede{f}{\bpi} \right) \mbox{d}s
\nonumber\\
&-\int_0^L\bOm\cdot\left(
\dede{f}{\bGam}\times\dede{g}{\bmu}
- \dede{g}{\bGam}\times\dede{f}{\bmu}
\right)\mbox{d}s
\nonumber\\
&
-\int_0^L
\bGam\cdot\left(
\dede{f}{\bGam}\times\dede{g}{\bpi}
-
\dede{g}{\bGam}\times\dede{f}{\bpi}
\right)\mbox{d}s+ \int_0^L \left( \dede{f}{\bGam}\cdot\partial_s\dede{g}{\bmu}- \dede{g}{\bGam}\cdot\partial_s\dede{f}{\bmu} \right) \mbox{d}s,
\end{align*}
\begin{equation}\label{bracket_right}
\begin{aligned}
\{f,g\}_R=& \int_0^L \nu \left(\frac{\partial g}{\partial \nu} \partial _s \frac{\partial f}{\partial \nu}  - \frac{\partial f}{\partial \nu} \partial _s \frac{\partial g}{\partial \nu}  \right)\mbox{d}s  + \int_0^L \xi\left( \frac{\partial g}{\partial \nu}\partial _s \frac{\partial f}{\partial  \xi }-\frac{\partial f}{\partial \nu}\partial _s \frac{\partial g}{\partial  \xi } \right)\mbox{d}s \\
&+\int_0^L S\partial _s \left(  \frac{\partial f}{\partial  S } \frac{\partial g}{\partial  \nu}-  \frac{\partial f}{\partial  \nu}  \frac{\partial g}{\partial  S } \right) \mbox{d}s.
\end{aligned}
\end{equation}
This Poisson bracket can be also directly derived by using the relations \eqref{dual_var_def} in equations \eqref{system_ell} and computing the time derivative of an arbitrary function $f$ depending on the variables $ \bpi$, $ \bmu$, $ \nu$, $ \bOm$, $ \bGam$, $ S$ and $ \xi$. The variable $R$, which is not symmetry-reduced, can be easily included, leading to an additional canonical Poisson bracket in the variable $R$ and its corresponding momentum $p_R= \dede{\ell}{\dot R}$.

\rem{ 
\section{One dimensional reductions and exact solutions}\label{sec_1D} 
\label{sec:1D} 

\todo{FGB: This section 4, is not needed I think, since we repeat this later.}
Let us now consider the propagation of 1D shock waves along the tubes without rotation. We need to assume $\Lambda={\rm Id}$ and $ \mathbf{r} =r \mathbf{E} _3$, so $\bom=\mathbf{0}$, $\bOm=\mathbf{0}$, $\bgam= \dot r \mathbf{E}_3$, and $\bGam= r' \mathbf{E}_3$. This allows to drop angular momentum equation completely, reduce the compatibility conditions to 1D scalar equation  
\begin{equation} 
\label{compatibility_1D} 
\gamma_s = \Gamma_t \, , 
\end{equation} 
and reduce the linear momentum equation of \eqref{full_3D} to a scalar 1D equation by projecting on $\mathbf{E}_1$ axis. We then take all the dependent variables to be functions on the combination $x=s-ct$ and express the system \eqref{full_3D} as a system of coupled ODEs. The propagation speed of the shock wave $c$ and the shape of the tube to be observed will then have to be found as a solution of the boundary value problem for the ODEs allowing for the boundary conditions going to fixed points as $x \rightarrow \pm \infty$. Note that these limiting points will have to be different since the pressure before and after the shock wave is different, and leading to different limiting radii before and after the shock wave. 
}

\section{Particular solutions for inextensible unshearable tubes, and comparison with previous works} 
\label{sec:1D_shock}

In this section we shall specify our model to the case of an inextensible and unshearable expandable tube. We then focus on straight expandable tubes with no rotational motion and compare our model with previous works. This is the case that has been studied extensively in the literature, mostly in the context of blood flow involving incompressible fluid, as we show later. We also show that these simplified models arise from a variational principle, which can then be used to derive the incompressible models by the Lagrange multiplier approach.

\subsection{Equations of motion for inextensible and unshearable tubes}

The dynamics of an inextensible and unshearable, but expandable, tube conveying compressible fluid is obtained by imposing the constraint $\bGam (t,s)= \bchi$, for all $t,s$. If we denote by $\mathbf{z}$ the Lagrange multiplier associated to this constraint and add the term $\int_0^T\int_0^L \mathbf{z} \cdot ( \boldsymbol{\Gamma} - \bchi)\mbox{d}s\,\mbox{d}t$
to the action functional in our variational principle \eqref{min_action_gas}, the first two equations in \eqref{full_3D} will change to
\[
\left\lbrace
\!\!\begin{array}{l}
\displaystyle\vspace{0.2cm}\frac{D}{Dt} \pp{\ell_0}{\bom}+\!\bgam\!\times\pp{\ell_0}{\bgam} +\frac{D}{Ds}  \left( \pp{\ell_0}{\bOm} +  (p-p_{\rm ext})  \pp{Q}{\bOm}  \right)   +\!\bGam\!\times 
 \left( 
\pp{\ell_0}{\bGam} + (p-p_{\rm ext})  \pp{Q}{\bGam} + \mathbf{z} 
\right) =0
\\
\displaystyle\frac{D}{Dt}\pp{\ell_0}{\bgam} + \frac{D}{Ds}\left( \pp{\ell_0}{\bGam}+ (p-p_{\rm ext})  \pp{Q}{\bGam}   + \mathbf{z}\right)=0.
\end{array}\right.
\]
The physical meaning of $\mathbf{z}$ is the reaction force enforcing the inextensibility constraint. 

Particular simple solutions of this system can be obtained by assuming the axis of the tube being straight and no rotational motion. In that case, the inextensibility constraint leads to $\mathbf{z} = z \bchi$, the direction along the axis. 
For such particular solutions, the angular momentum equation is satisfied identically, and the linear momentum equation reduces to a one dimensional equation for the reaction force $z$. Therefore, we only need to compute the fluid momentum equation and the Euler-Lagrange equations for $R$. From the third and fourth equations in \eqref{full_3D}, we get
\begin{equation}\label{inext_unshear_norot}  
\left\lbrace\begin{array}{l}
\displaystyle\vspace{0.2cm}\partial _t \big( \xi u \big)+\partial _s \big( \xi u^2 + p A \big) = p \partial _s A \\
\displaystyle a  \ddot R- \partial _s  \pp{F}{R'} + \pp{F}{R} = 2 \pi R \left( p - p_{\rm ext} \right),
\end{array} \right.
\end{equation} 
together with the conservation of mass and entropy 
\begin{equation} 
\partial _t \xi +\partial _s( \xi u)=0, \quad \partial _t S +u \partial _s S=0 ,
\label{xi_S_cons}
\end{equation}
and where we assumed $A(R)=\pi  R^2 $ and that $ \mathbb{I}  $, $ \lambda $, and $ \mathbb{J}  $ do not depend on $R$. This gives a system of four equations for the four variables $u$, $R$, $S$, and $ \xi$, where we recall that $ \rho = \xi /A$ and $p= \rho ^2 \frac{\partial e}{\partial \rho }$. Using mass conservation and $ \rho = \xi /A$, the fluid momentum equation can be rewritten as
\begin{equation}
\label{meq2}
\partial _t u+ u\partial _s u =-\frac{1}{\rho} \partial_s p,
\end{equation}
which has formally the same form with that of 1D compressible fluids. Note that the equation governing the evolution of $\rho$ is different from the corresponding 1D compressible fluid case.

\medskip

Note that the system of equation we obtained in \eqref{inext_unshear_norot}--\eqref{xi_S_cons} has been directly deduced by making the assumption of a straight tube in the equations \eqref{full_3D}. We shall show below, that the resulting system is itself a Lagrangian system arising from Hamilton's principle.
Indeed, an alternative way to think about the case of inextensible and unshearable, but expandable, straight tube is to consider a Lagrangian $\ell$ that depends only on $R$ and $u$ for incompressible flows, and additionally on the thermodynamic variables for compressible flows. We shall consider these two cases separately for clarity of comparison with the previous literature. In that case, $s=x_1$ takes the meaning of the Euclidian coordinate along the axis, with all deformation variables related to the tube being trivial, \emph{i.e.}, $\bgam=\mathbf{0}$, $\bGam=\bchi$, $\bom=\mathbf{0}$ and $\bOm=\mathbf{0}$. 

One may question about the physicality of the assumption $s=x_1$. Indeed, under stretching of the tube along the axis, the material point situated at a given point $s$ will  move along the axis as well which lead to the dependence $\mathbf{r}(t,s)$. Thus, the dynamics moves material points in the direction of the  axis $\bchi=\mathbf{E}_1$. If one were to insist to consider the evolution of the dynamics for a \emph{given} value of the coordinate $x_1$, it would lead to the spatial representation of the elastic model of the tube. Physically, this would give the description of the tube's motion in fixed spatial coordinates, similar to Eulerian description of fluid's motion. Unlike the fluid case, however, describing elastic motion in Eulerian framework  is known to be notoriously difficult from both mathematical and physical considerations \cite{MaHu1994}. 

However, if one assumes both inextensibility of the tube and no elastic deformations of the tube's axis, then the description of elastic motion of the walls in our method, \emph{i.e.}, the convective elastic frame coincides with the spatial frame. All previous works on the subject have implicitly made this assumption. Thus, we will spend some time investigating the inextensible and unshearable tube case, both in order to connect with previous works, and also to make analytic progress for the case of traveling wave solutions. 

\paragraph{Lagrangian variational formulation for compressible fluids in expandable straight tubes.}
We consider the Lagrangian given by
\begin{equation} 
\ell(u,\xi , S, R , \dot R)=\int_0^L \left[ \frac{1}{2} \xi u^2 +\frac{1}{2} a \dot R ^2 - \xi e \left( \frac{ \xi }{ A(R) } ,S\right) -F(R,R' ) \right] {\rm d}s.
\label{lagr_compressible_simple} 
\end{equation}
The variational principle is written as 
\begin{equation}\label{VP_unshear_inext}
\delta \int_0^T \ell(u, \xi , S, R, \dot R)dt=0,
\end{equation}
for  variations $\delta u =\partial _t\eta + u\partial _s \eta -\eta \partial_s u$, $\delta \xi = - \partial _s ( \xi\eta )$, $\delta S= -\eta\partial_sS$, where $\eta(t,s)$ vanishes at $t=0,T$ and for variations $\delta R$ vanishing at $t=0,T$. This variational principle is obtained, exactly as earlier, by rewriting the classical Hamilton principle
\[
\delta \int_0^T \mathsf{L}(\varphi, \dot\varphi, R, \dot R){\rm d}t=0
\]
in terms of the Eulerian variables $u=\dot\varphi\circ\varphi^{-1}$, $\xi= (\xi_0\circ \varphi^{-1})\partial_s\varphi^{-1}$, $S=S_0\circ\varphi^{-1}$ and computing the variations $\delta u$, $\delta\xi$, $\delta S$, induced by $\delta\varphi$.
 {By applying \eqref{VP_unshear_inext} and collecting terms proportional to $\eta$ 
and $\de R$ we get the system  \eqref{inext_unshear_norot}. Equations in \eqref{xi_S_cons} are deduced from the relations $\xi= (\xi_0\circ \varphi^{-1})\partial_s\varphi^{-1}$ and $S=S_0\circ\varphi^{-1}$.

By taking the Legendre transform $(u,\xi, S, R,\dot R)\mapsto (\nu, \xi, S, R, p_R)$, with $\nu=\dede{\ell}{u}$ and $p_R=\dede{\ell}{\dot R}$ we immediately obtain that the system \eqref{inext_unshear_norot}--\eqref{xi_S_cons} is Hamiltonian with respect to the Poisson bracket 
\[
\{f,g\}_R+ \{f,g\}_{\rm can}
\]
where the term is given in \eqref{bracket_right} and the second term is the canonical Poisson bracket in $R$ and $p_R$.

\rem{ 
Equations: abstract
\[
\left\{
\begin{array}{l}
\vspace{0.2cm}\displaystyle\partial _t \frac{\delta \ell}{\delta u} +\partial _s\left(\frac{\delta \ell}{\delta u}u \right)+  \frac{\delta \ell}{\delta u}\partial _s u= \xi \partial_s \frac{\delta \ell}{\delta \xi } - \frac{\delta \ell}{\delta S} \partial _s S\\
\vspace{0.2cm}\displaystyle\partial _t\xi +\partial _s (  \xi u)=0, \qquad \partial _t S+ u \partial _s S=0\\
\displaystyle\partial _t\frac{\delta \ell}{\delta \dot R}- \frac{\delta \ell}{\delta  R} =0
\end{array}
\right.
\]
for $\ell$ above,
} 

\paragraph{Incompressible fluids in expandable straight tubes.} In the incompressible case, we shall only treat the case when thermodynamic effects may be neglected. In general, thermodynamic effects for the incompressible fluid may be considered by using the Lagrangian \eqref{lagr_compressible_simple} with an additional incompressibility conditions. Such considerations may be important later for the cases when we would want to introduce friction and thermal effects. However, for most applications, such as engineering and blood flows, thermodynamic effects are negligible. We will thus consider $\rho=\rho_0=const$ and the Lagrangian in the form 
\begin{equation} 
\label{lagr_incompressible_simple}
\ell(u, \xi, R , \dot R)=\int_0^L \left[ \frac{1}{2} \xi  u^2 +\frac{1}{2} a \dot R ^2  -F(R,R') \right]  {\rm d}s\,. 
\end{equation}
The incompressibility constraint is included via the Lagrange multiplier $\mu$ in a similar way with \S\ref{sec_incomp}, thus yielding the variational principle
\begin{equation} 
\delta \int_0^T \left[  \ell(u, \xi , R, \dot R)+ \int_0^L \mu \left( A(R)- (A_0 \circ \varphi ^{-1} )\partial _s \varphi^{-1} \right) {\rm d}s\right] {\rm d}t=0 \,,
\label{variations_incompressible_simple} 
\end{equation}
with variations $\delta \varphi$ and $\delta R$ vanishing at $t=0,T$ and variations $\delta\mu$, where $u=\dot\varphi\circ\varphi^{-1}$ and $\xi= (\xi_0\circ\varphi^{-1})\partial_s\varphi^{-1}$. In applying \eqref{variations_incompressible_simple}, the terms proportional to $\eta$, $\delta R$, and $\delta \mu$ give, respectively, the fluid momentum conservation, the equation for the radius, and the incompressibility condition which we differentiate once with respect to time for convenience. We get the system
\begin{equation}\label{incompressible_system}
\left\{
\begin{array}{l}
\vspace{0.2cm}\displaystyle\partial _t ( \rho_0 Au)+\partial _s ( \rho_0 Au ^2 )=- A\partial_s \mu \\
\vspace{0.2cm}\displaystyle a \ddot R - \partial _s \frac{\partial F}{\partial R'} + \frac{\partial F}{\partial R} = 2 \pi R  (\mu-p_{\rm ext})
\\
\displaystyle\partial _t A+\partial _s (  A u)=0\, ,
\end{array}
\right.
\end{equation} 
where we recall that $\xi=\rho_0 A$, $ A(R)= \pi R ^2 $, and $\rho_0 =const$.

\subsection{Comparison with previous works} 
\label{sec:comparison} 
\paragraph{Motion of a compressible fluid in a nonuniform tube with rigid walls.} For a \emph{prescribed} cross-section $A(s)$, the compressible system \eqref{inext_unshear_norot}--\eqref{xi_S_cons} reduces to
\begin{equation}
\left\{
\begin{array}{l}
\vspace{0.2cm}\displaystyle\partial _t ( \rho Au)+\partial _s ( \rho Au ^2 )=- A\partial_s p\\
\vspace{0.2cm}\displaystyle\partial _t ( \rho A)+\partial _s ( \rho A u)=0\\ 
\partial _t S+ u \partial _s S=0,
\end{array}
\right.
\label{prescribed_A} 
\end{equation} 
since the second equation in \eqref{inext_unshear_norot} can be discarded for rigid walls.
The first two equations can be written in the form 
\begin{equation} 
\label{Whitham_eqs8_12} 
\left\{
\begin{array}{l}
\vspace{0.2cm} 
\displaystyle\partial _t \rho + u \partial _s\rho + \rho \partial _s u + \rho u  \frac{\partial _sA}{A}=0  \\ 
\partial _t u + u \partial _s  u + \frac{\displaystyle 1}{\displaystyle \rho} \partial _s p =0 
\end{array}
\right.
\end{equation} 
which are exactly the equations (8.1) and (8.2) of \cite{Wi1974} describing the motion of gas in a tube with prescribed cross-section. Moreover, we can notice that since $p=p(\rho,S)$, we have 
\[ 
\partial _t p + u \partial _s p = \pp{p}{\rho} \lp  \partial _t \rho + u  \partial _s\rho\rp + \pp{p}{S} \lp \partial _t S + u \partial _s S \rp \, . 
\] 
Thus, the third equation of \eqref{prescribed_A} is equivalent to
\begin{equation} 
\label{Whitham_eqs8_3}
\partial _t p + u \partial _s p = c^2 \lp  \partial _t \rho + u  \partial _s\rho\rp , 
\end{equation} 
which is exactly equation (8.3) of \cite{Wi1974}.

\paragraph{Fluid motion in straight tubes with deformable walls.}
We are not aware of any work of compressible fluid inside a tube with stretchable walls, as described in \eqref{inext_unshear_norot}--\eqref{xi_S_cons}. We thus study the connection of its incompressible version, \emph{i.e.}, system 
\eqref{incompressible_system}, with the established models for fluid motion inside a tube with stretchable walls \cite{PeLu1998}, see also \cite{KoMa1999,JuHe2007,StWaJe2009,Tang-etal-2009,HeHa2011}, further summarized in review \cite{Se2016},
\begin{equation}\label{collapsible_models}
\left\{
\begin{array}{l}
\vspace{0.2cm}\displaystyle  \rho_0(\partial _t u +u \partial _s u )=- \partial_s p - \tau (u,A)\\
\vspace{0.2cm}\partial _t A+ \partial _s (Au)=0\\
\displaystyle p-p_{\rm ext}= \Phi (A)- T \partial _{ss} A,
\end{array}
\right.
\end{equation} 
where $ \rho _0$ is a constant, $p_{\rm ext}$ is the external pressure,  the function $ \Phi(A) $ is called the tube law and is such that $\Phi (A_0)=0$ with $A_0$ the reference cross-section, and $T\partial _{ss} A$ is a longitudinal stretching. The first two equations of this system coincide with the fluid momentum and incompressibility condition of the system \eqref{incompressible_system}. The term $\tau(u,A)$ describes the friction in the tube and is neglected in our model. Examples of $ \Phi(A) $ encountered in the literature are $\Phi (A)= \beta \big(\sqrt{A}-\sqrt{A_0}\big)/\sqrt{A_0}$ or $ \Phi (A)= \beta \big( \left(A/A_0 \right) ^{ \beta _1}-1\big)$. 
In our case, the last equation in \eqref{collapsible_models} is recovered from the second equation in \eqref{incompressible_system} by a particular choice of $F(R,R')$. 
Indeed, with
\begin{equation}\label{Pedley_Luo_F} 
F(R, R')= A(R) p_{\rm ext}+  \frac{T}{2} \big( \partial _s (A(R)) \big)^2    +f(A(R)),
\end{equation} 
where the function $f(A)$ is such that $ \frac{\partial f}{\partial A}= \Phi (A)$, our equation (with $a=0$) gives
\begin{align*} 
2 \pi R  p&=2 \pi Rp_{\rm ext}+\frac{\partial F}{\partial R} -  \partial _s \frac{\partial F}{\partial R'}  = 2 \pi Rp_{\rm ext}+2\pi R\frac{\partial f}{\partial A}- (2 \pi ) ^2 RT( (R') ^2 +RR'')\\
&=  2 \pi Rp_{\rm ext}+2\pi R\frac{\partial f}{\partial A}- 2\pi R T \partial_{ss}( \pi R ^2 )= 2 \pi R \left( p_{\rm ext}+ \Phi (A)- T \partial _{ss}A \right) ,
\end{align*} 
which coincides with the last equation in \eqref{collapsible_models}. 

\paragraph{Propagation of pressure disturbances in the tube.} 
The question of pressure pulse propagation through arteries is an important problem that has received considerable attention in the literature. The easiest way to derive this equation is to start with a system of mass and momentum conservation laws for the fluid in the following form \cite{Se2016}
\begin{equation}\label{pulse_eq} 
\left\{
\begin{array}{l}
\displaystyle\vspace{0.2cm}\partial_t A+ \partial_s (A u) =0 
\\ 
\displaystyle\partial_t u+ u \partial_su + \frac{1}{\rho_0} \partial_s p =0 \,, \quad \rho_0=const.
\end{array}
\right.
\end{equation} 
This system arises from the first and last equation of our system \eqref{incompressible_system} taking $\mu=p$.
Following \cite{Se2016}, one would take $A=A(p)$, assuming a given function of compliance of the wall with respect to internal pressure. This assumption leads to a closed system of equations for pressure and velocity 
\begin{equation}\label{pulse_eq} 
\left\{
\begin{array}{l}
\displaystyle\vspace{0.2cm} A'(p) \partial_t p + \partial_s \lp A(p) u\rp =0 
\\ 
\displaystyle\partial_t u + u \partial_su + \frac{1}{\rho_0} \partial_s p =0 \,, \quad \rho_0=const. 
\end{array}
\right.
\end{equation} 
A wave equation for the velocities can be obtained in the limit of small deviations from the equilibrium values, see (50) in \cite{Se2016}.

In our system \eqref{incompressible_system} treating the incompressible case, the pressure-like term $\mu$ comes as the Lagrange multiplier for incompressibility. Nevertheless, our system \eqref{incompressible_system} does recover \eqref{pulse_eq}. Indeed, suppose that in the second equation of \eqref{incompressible_system} the elastic function $F$ is only dependent on $R$ and not on its derivatives, and that the inertia terms proportional to $\ddot R$-equation are neglected ($a=0$). Then, the second equation in \eqref{incompressible_system} yields an algebraic relation $\mu=\mu(R)$, or, alternatively, $A=A(\mu)$. The rest of the derivation follows by the substitution of $A=A(\mu)$ into the first and third equation of \eqref{incompressible_system}.

\paragraph{Models involving wall inertia.}
More complex models which involve terms with time derivatives of $R$ in the equation for the radius have the form \cite{QuTuVe2000,FoLaQu2003}: 
\begin{equation} 
\label{inertia_R}
\left\{
\begin{array}{l}
\displaystyle\vspace{0.2cm}\partial_t G + \partial_s \left( \alpha \frac{G^2}{A} \right) + \frac{A}{\rho} \partial_s p + K \frac{G}{A} =0
\\ 
\displaystyle\vspace{0.2cm}\alpha  \frac{\partial^2 R }{\partial t^2} -\gamma _1 \pp{  R}{t} - a  \frac{\partial^2R }{\partial s^2} - c \frac{\partial^3R }{\partial s^2 \partial t} + b R = p-p_{\rm ext}
\\ 
\displaystyle\partial_t A + \partial_s G =0 ,
\end{array}
\right.
\end{equation}
where $G$ is the flux through the cross-section of the tube and $\alpha$, $a$, $c$, $K$ are constant coefficients. In our approximation of inviscid plug flow, $G=A u$, equations \eqref{inertia_R}  have a dimensionless coefficient of order $1$ in front of the term $\partial_s (A u)$ since the averaging is done based on the Poiseuille flow profile rather than inviscid plug flow treated in our paper. 
Keeping that in mind, the third equation of continuity  coincides with ours, apart from that dimensionless coefficient. The first equation of \eqref{inertia_R} coincides with the first equation of our system 
\eqref{incompressible_system}, apart from the last term proportional to $K$. That term comes from the fluid friction and is not present in our theory. Note however, that the existence of such a term for describing friction can be introduced in the variational network as well \cite{FGBPu2015}, although we postpone this discussion for further studies. 
The term proportional to $\gamma$ in the second equation of \eqref{inertia_R} comes from the dissipative motion of the tube's wall and is not present in our variational model, although it could be introduced using the Lagrange-d'Alembert principle to incorporate external forces. Although this dissipation term is important for arterial flow applications, we shall not consider  it as our primary focus is the tube in air where the dissipation of the motion of tube's wall can be neglected. A complete variational formulation including all these effects as irreversible processes can be developed by combining the approach in this paper with that of \cite{FGBYo2017a,FGBYo2017b}. The terms proportional to the coefficients $a$ can be incorporated in our theory by a particular choice of $F(R,R')$. The term proportional to the coefficient $c$ describes viscoelastic effects essential for arteries flow and is not present in our theory. We shall note that more complex equations for the radius $R$, for example, involving terms proportional to $\partial^4 R/\partial s^4$ can be naturally incorporated into our system if the elastic energy of wall deformation is allowed to depend on $R''$, \emph{i.e.}, $F=F(R,R',R'')$. Finally, we note that the more general system \eqref{full_3D_incompressible}  provides a rigorous treatment of the longitudinal deformations of the rod. 

\rem{ 
Our model gives
\[
p-p_{\rm ext}= \Phi (A)- T \partial _{ss} A+ \frac{1}{2 \pi R} \left( a R \ddot R+ \frac{1}{2} a\dot R ^2  \right) .
\]

\medskip

Before going to the full Cosserat case, we can think about the analogue of the form $A( \boldsymbol{\Omega} , \boldsymbol{\Gamma} )= A_0- \frac{K_{\boldsymbol{\Gamma} }}{2}|\boldsymbol{\Omega} | ^2 -D_{ \boldsymbol{\Gamma} } \mathbf{E} _1 \cdot (\boldsymbol{\Gamma} -\mathbf{E} _1)- \frac{K_{ \boldsymbol{\Gamma} }}{2}|\boldsymbol{\Gamma} -\mathbf{E} _1| ^2 $ for the unshearable, but extensible and expandable straight tube. With $ \Lambda = Id$ and $ \mathbf{r} =(r',0,0)$, this gives
\[
A(R,r)=\pi  R^2 \Big( 1- D_{\boldsymbol{\Gamma} }(r'-1)- \frac{K_{\boldsymbol{\Gamma} }}{2} (r'-1) ^2 \Big).
\]
} 

\subsection{Traveling wave solutions and numerical simulations} 
\label{sec:traveling_waves}
Let us now study the solutions representing nonlinear traveling waves, \emph{i.e.} solutions of  \eqref{inext_unshear_norot}--\eqref{xi_S_cons} of the form $f(t,s)=f(s-ct)$ where $c$ is the speed of the wave to be determined. In what follows, we denote the derivative with respect to the argument $x=s-ct$ by primes.  Conservation of entropy \eqref{xi_S_cons} gives $S'=0$ or $S=S_0$ away from the shock, so the dynamics is isentropic. Therefore, we will denote 
$p=p(\rho)$ for a given $S=S_0$. 

The $\xi$-equation of  \eqref{xi_S_cons} leads to the integrated conservation of mass written as 
\begin{equation} 
\pi R^2 \rho (u-c) = \pi R_\pm ^2 \rho _\pm(u_\pm-c)\,,
\label{Q_cons_wave}
\end{equation} 
where $ \rho_\pm$, $u_\pm$, $R_\pm$ are the values at $x=\pm\infty$.

The first equation in \eqref{inext_unshear_norot} can be combined with the conservation of mass to provide an integral of motion \textit{away from the shock}: 
\begin{equation} 
\label{meq2_wave}
 u \Big( \frac{1}{2}  u-c\Big)  + h(\rho, S_\pm) = u_\pm\Big( \frac{1}{2}  u_\pm-c\Big)  + h(\rho_\pm,S_\pm)\,,
\end{equation}
where $h(\rho,S)=p/\rho+e$ is the enthalpy.
For given $R$ and $c$, equations \eqref{Q_cons_wave} and \eqref{meq2_wave} give two algebraic equations for the two unknowns $u$ and $\rho$, denoted $u=u(R;c)$ and $\rho=\rho(R;c)$.
The second equation in \eqref{inext_unshear_norot} reduces to 
\begin{equation} 
\label{Req_wave} 
\!\!\!\!- a c^2   R'' +\frac{d}{d x}\pp{F}{R'} - \pp{F}{R}+ 2 \pi R\, \left( p(\rho,  S_\pm) -p_{\rm ext} \right) =0\,,
\end{equation} 
with $R'(x) \rightarrow 0$ when $x \rightarrow \pm\infty$.

\medskip

\noindent \textit{Numerical solutions.}
To illustrate our method, we take an undisturbed tube which is made out of 1mm soft rubber wall with inner radius $ R_-=5$\,cm, wall thickness $h=1$\,mm, Young modulus $E=5$\,MPa, and Poisson ratio $\sigma=0.5$, giving shear modulus $G=1.66$\,MPa through $E=2 G (1+\sigma )$. The  density of the rubber is  $\rho_s=1522$\,kg/m$^3$ and the gas in its initial state $ x \rightarrow - \infty$ is at rest ($ u_-=0$), and at atmospheric pressure with room temperature, corresponding to the density  $\rho=1.225$\,kg/m$^3$. We assume that the gas is perfect with adiabatic coefficient $\gamma=C_p/C_v=1.4$ and 
equation of state $p= (\gamma-1) \rho e$,  with the external pressure $p_- = p_{\rm ext}= 1$\,atm. 
We use the simplest elastic energy $2 F(R,R') =  \beta_1 R'^2+ \beta_2 (R-R_-)^2$ \cite{QuTuVe2000} with:
\vspace{-1.5mm}
\begin{equation} 
\label{beta_1_2} 
 \beta_1=  2 \pi  K G h R_- \, , \qquad 
 \beta_2 = \frac{ 2 \pi   }{R_-} \frac{E h}{1- \sigma^2}  \, , \qquad a = 2 \pi R_- \rho_s h, 
\end{equation}
with $K$ the Timoshenko coefficient, $\beta_1 \simeq 1.14\cdot 10^5$\,N/m, $\beta_2 \simeq 8.37 \cdot 10^5$\,N$/ m^3$, 
$a \simeq 0.4782$\,kg/m.
We are looking for a solution $R=R(x)$ having a shock satisfying the Rankine-Hugoniot conditions.

\rem{ 
\begin{equation} 
\begin{aligned} 
& \left[ \rho (u-c) \right]=0\, , \quad  \left[ \rho (u-c)^2  | \bGam |^2  + \rho(u-c) c \bGam \cdot \Lambda^{-1} \mathbf{E}_1 + p \right]  =0 
\\ 
& \big[  \rho e  (u-c)  + \frac{1}{2}  \rho (u-c) \left( c^2 + 2 (u-c) c \bGam \cdot \Lambda^{-1}\mathbf{E}_1 + (u-c)^2 |\bGam|^2  \right) +   \\ 
& \qquad \qquad \qquad \qquad 
\frac{p}{|\bGam|^2} \left( |\bGam|^2 (u-c) + c \bGam \cdot \Lambda^{-1}\mathbf{E}_1  \right) \big] = 0 
\label{RH_moving_1} 
\end{aligned} 
\end{equation} 
} 
We are looking for a solution that is smooth for $x>0$ and $x<0$, tends to  steady state $R=R_-$ as $R \rightarrow -\infty$ and also tends to a fixed state $R=R_+$ as $x \rightarrow + \infty$. Additionally, the solution $R(x)$  and $R'(x)$ is continuous  at the shock, whereas the variable $u(x)$, $\rho(x)$, and $S(x)$ have a jump at the shock satisfying Rankine-Hugoniot conditions. For a chosen value of $R=R_s$ at the shock, such solution will only exist for particular  value of $c$, also yielding the limiting value $R=R_+$ that is dependent on the choice of $R_s$. As is shown in \cite{Wi1974}, the jump of entropy  across the shock is always positive, $S_{+}-S_{-}>0$. On Figure~\ref{fig:solution_tube}, we present a solution computed at $c \simeq 447$\,m/s, with the limiting pressure behind the shock wave being $p_+\simeq 1.826$\,atm. The solution is presented in the dimensionless coordinates $x/R_-$ and $(y,z)/R_-$. 

We shall also note that physically, one expects monotonic solution in $R(x)$. From \eqref{Req_wave}, such solutions only exist 
for $c<c_*\sqrt{\beta_1/a} \simeq 488$\,m/s. 
\begin{figure}[h]
\centering 
\includegraphics[width=0.8\textwidth]{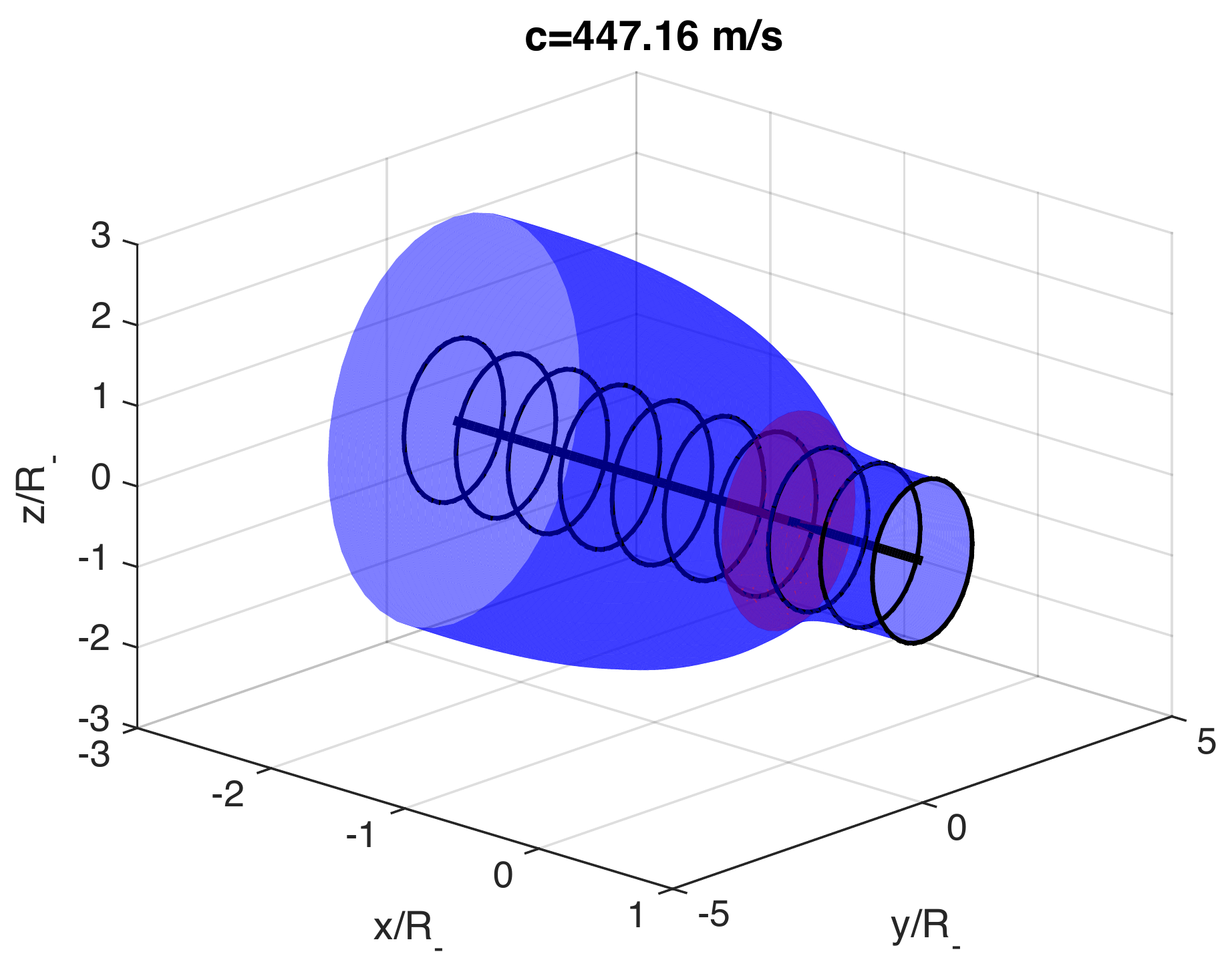} 
\caption{A shock propagating along a tube. 
Propagating shock is shown in red, moving in the direction of positive  $x$ with velocity $c$, \emph{i.e.} $s=x -ct$. The position  of the shock is chosen to be at $x=0$. The pressure is increasing towards larger $x$, reaching the limiting value 1.965 atm. The speed of the shock is computed to be $c \sim 447$ m/s. 
The centerline of the tube is shown with a solid black line, and the undisturbed position of the cross-section of the tube are shown by solid circles. 
\label{fig:solution_tube}
}  
\end{figure}

It is clear that the tube plays an important role in the propagation of the shock, particularly on shock speed. In order to investigate this effect further, we study the dependence of the shock's speed on the effective strength of the shock which we define as $z=(p_+-p_-)/p_-$. For experimental realization of the flow inside the tube,  $p_-$ and $p_+$ which are pressures at the end of the tube, are variable parameters, and the strength of the shock itself at the tube is computed subject to setting these parameters.  This is in contrast to the classic 1D shocks where specifying the pressures immediately before and after the shock determine the result for the whole system. We plot the Mach number $ M= \frac{c-u_-}{{c_s}_-}$ defined as the ratio of shock's speed to the sound speed $ {c_s}_- ^2 =  \gamma \frac{p_-}{ \rho _-} $ in front $x \rightarrow - \infty$ of the shock.
The results are summarized on Figure~\ref{fig:ShockSpeeds}, plotted with a solid line. On the same Figure, we also plot the Mach number computed from the classical formula $M=\left( 1+ \frac{1+\gamma}{2 \gamma} z \right)^{1/2}$ from \cite{Wi1974}. Note the discrepancy between the Mach numbers for the classic shock and our case which is purely due to the effect of the tube. 
\begin{figure}[h]
\centering 
\includegraphics[width=0.8 \textwidth]{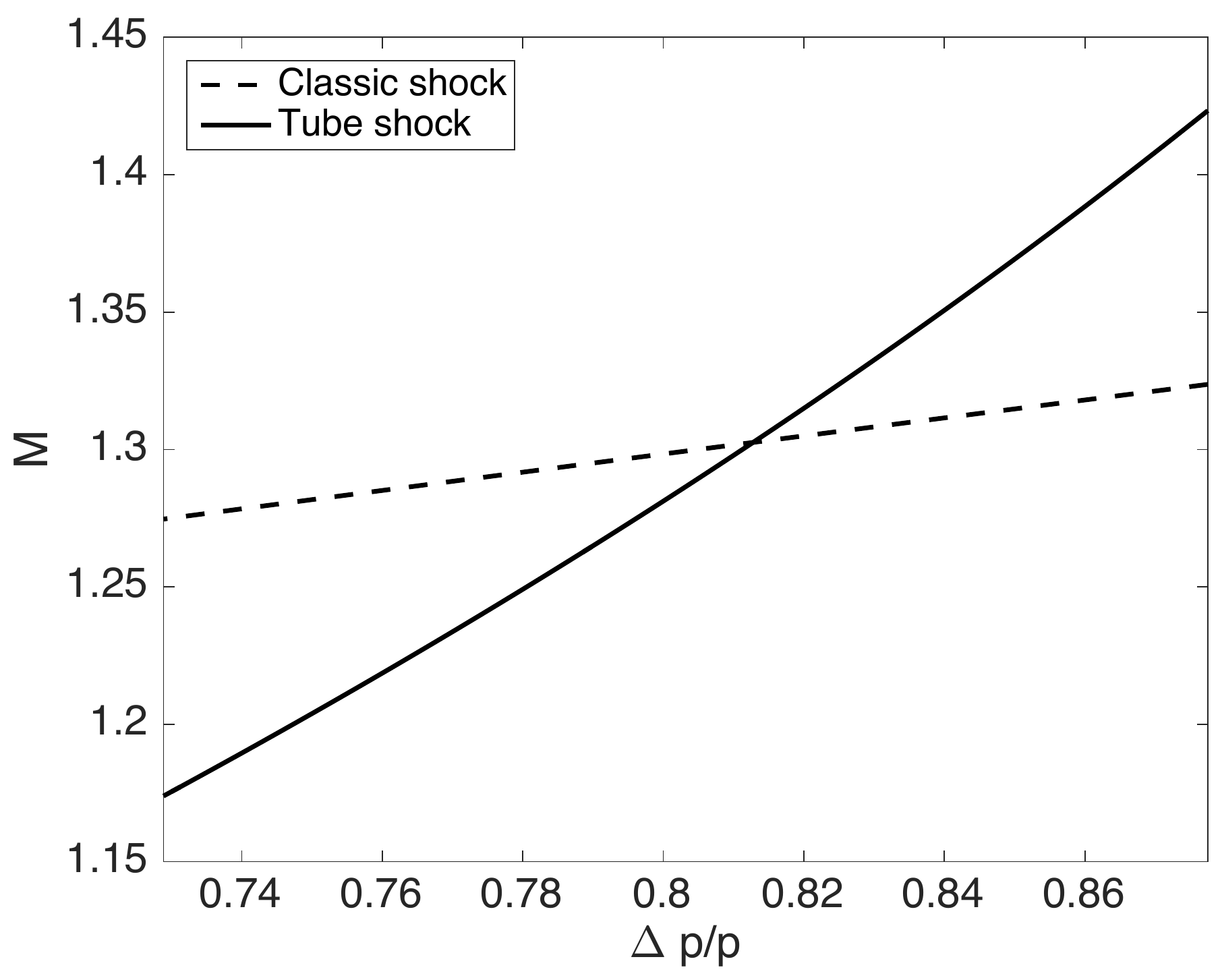} 
\caption{
Mach number of the shock as a function of shock's strength $  z=\Delta p/p$ (solid line). We also present the formula for the Mach number from \cite{Wi1974} (dashed line). 
\label{fig:ShockSpeeds}
}  
\end{figure}

\section{Conclusions}
In this paper, we have derived the general equations of motion for a tube with compliant walls conveying an inviscid fluid which may be compressible or incompressible. The system is capable to incorporate arbitrary three-dimensional deformations, treat discontinuities in the fluid flow such as shock waves, and incorporate arbitrary elasticity properties of the compliant walls. In the model presented here, the cross-section of the tube is assumed circular, as it is described by a single function $R(t,s)$. More generally, one can extend this model to include an elliptical cross-section of the tube, parameterized  by \emph{e.g.} its semi-axes $a(t,s)$ and $b(t,s)$, or even more complex shapes with several parameters characterizing it.  We shall note that such an extension is mathematically doable in our model, although the resulting equations become quite cumbersome. It is interesting to investigate, however, if these equations lead to an analytical description of the instability, and even local collapse when the cross-sectional area of the tube tends to $0$ at a particular value of $s$ and $t$. Such a finite-time singularity, if it exists, would be an interesting feature of the problem. 

Another interesting development of the model would be the inclusion of friction using a lubrication approach, as developed in \cite{Ch-etal-2018,ShCh2018}. It would then be natural to merge this friction term with the recently developed variational approach to continuum systems with irreversible processes \cite{FGBYo2017a,FGBYo2017b}. This approach should  be able to incorporate both the inertial and the frictional terms in the fluid, the tube, and the wall motion consistently in a single model. These and other interesting questions will be addressed in future work.

\section{Acknowledgements} 
We gratefully acknowledge discussions with Profs.~I.~C.~Christov, D.~D.~Holm, T.~W.~Secomb and P.~Vorobieff.  FGB is partially supported by the ANR project GEOMFLUID 14-CE23-0002-01. VP was partially supported by NSERC Discovery Grant and the University of Alberta.

{\footnotesize 
\bibliographystyle{unsrt}
\bibliography{Garden-hoses}
}

\end{document}